\documentclass[11pt, letterpaper]{amsart}
\usepackage{graphicx, amssymb, color}
\usepackage{amsmath}
\usepackage[round, sort&compress]{natbib}

\newtheorem{thm}{Theorem}[section]
\newtheorem{cor}[thm]{Corollary}
\newtheorem{lem}[thm]{Lemma}
\newtheorem{prop}[thm]{Proposition}
\theoremstyle{definition}

\newtheorem{ass}[thm]{Assumption}

\theoremstyle{remark}
\newtheorem{rem}[thm]{Remark}

\numberwithin{equation}{section}

\newcommand{\Real}{\mathbb R}

\newcommand{\F}{\mathcal{F}}

\newcommand{\prob}{\mathbb{P}}
\newcommand{\qprob}{\mathbb{Q}}

\newcommand{\expec}{\mathbb{E}}

\newcommand{\cX}{\mathcal{X}}

\newcommand{\cL}{\mathcal{L}}
\newcommand{\cadlag}{c\`adl\`ag}

\newcommand{\fR}{\mathfrak{R}}
\newcommand{\lR}{\underline{\mathfrak{R}}}
\newcommand{\uR}{\overline{\mathfrak{R}}}
\newcommand{\pX}{X^{(p)}}
\newcommand{\ptX}{\wt{X}^{(p)}}

\newcommand{\pr}{r^{(p)}}
\newcommand{\esssup}[1]{{\text{esssup}_{#1}}}

\newcommand{\indic}{\mathbb{I}}
\newcommand{\pare}[1]{\left(#1\right)}
\newcommand{\bra}[1]{\left[#1\right]}
\newcommand{\wt}[1]{\widetilde{#1}}

\title[Stability of the exponential utility maximization problem]{Stability of the exponential utility maximization problem with respect to preferences}

\author[]{Hao Xing}
\thanks{We would like to thank Constantinos Kardaras, Johannes Muhle-Karbe, and Marcel Nutz for helpful comments. We are also grateful to the associated editor and two anonymous referees for carefully reading this paper and providing valuable suggestions, which assisted us in improving this paper.
This research is supported in part by STICERD at London School of Economics.}
\address{Department of Statistics, London School of Economics and Political Science, 10 Houghton st, London, WC2A 2AE, UK}
\email{h.xing@lse.ac.uk}

\date{\today}

\begin{document}


\begin{abstract}
 This paper studies stability of the exponential utility maximization when there are small variations on agent's utility function. Two settings are considered. First, in a general semimartingale model where random endowments are present, a sequence of utilities defined on $\Real$ converges to the exponential utility. Under a uniform condition on their marginal utilities, convergence of value functions, optimal payoffs and optimal investment strategies are obtained, their rate of convergence are also determined. Stability of utility-based pricing is studied as an application. Second, a sequence of utilities defined on $\Real_+$ converges to the exponential utility after shifting and scaling.  Their associated optimal strategies, after appropriate scaling, converge to the optimal strategy for the exponential hedging problem. This complements Theorem 3.2 in \textit{M. Nutz, Probab. Theory Relat. Fields, 152, 2012}, which establishes the convergence for a sequence of power utilities.
\end{abstract}

\maketitle

\setcounter{section}{-1}

\section{Introduction}\label{sec: introduction}

This paper considers an optimal investment problem where an agent, whose preference is described by a utility function, seeks to maximize expected utility of her wealth from investment and a random endowment (illiquid asset) at an investment horizon $T\in \Real_+$. Given two problem primitives: utility function and market structure, the goal is to identify the optimal investment strategy that the agent undertakes. When the utility has constant absolute risk aversion, \cite{six-authors} give an elegant solution to this problem. We study in this paper stability of the optimal investment strategy when agent's utility deviates from exponential utility. In particular, we are interested in a \emph{quantitative} measure on how far the optimal strategy deviates when there are small variations on agent's preference.

Two settings are studied. First, consider a sequence of utility functions $(U_\delta)_{\delta>0}$, each of which is defined on $\Real$, such that it converges pointwise to $U_0$ which has unit absolute risk aversion. Deviation is measured by two components: i) the ratio of marginal utilities $\fR_\delta$ between $U_\delta$ and an exponential utility $\wt{U}_\delta$ with absolute risk aversion $\alpha_\delta$; ii) the difference between $\alpha_\delta$ and $1$. The first component measures how far $U_\delta$ is away from an exponential utility; while the second component describes how far this exponential utility is away from the exponential utility with unit risk aversion. When $\fR_\delta$ is bounded from above and away from zero, uniformly in $\delta$, our first main result, Theorem \ref{thm: conv R}, states the convergence of the optimal payoffs and value functions in a general semimartingale model; moreover the convergence of optimal strategies also follows, when asset price processes are continuous. Beyond these continuity results, the rate of convergence is determined in Corollary \ref{cor: convergence rate}. Aforementioned two components of variations impact deviation of the optimal payoff (hence the optimal strategy) at different rates: the convergence of absolute risk aversions has \emph{first order} impact, while the convergence of $\fR_\delta$ has \emph{second order} effect. Stability of utility based prices, Davis price and indifference price, with respect to agent's preference is also discussed as an application; cf. Corollaries \ref{cor: davis price} and \ref{cor: indifference price}.

The stability problem studied in the first setting is similar to \cite{Carassus-Rasonyi}, where the problem is formulated in a discrete time setting and asset price processes are assumed to be bounded. For utilities defined on $\Real_+$, aforementioned stability problem has been extensively studied. \cite{Jouini-Napp} consider an It\^{o} process model. \cite{Larsen} extends the analysis to continuous semimartingale models. \cite{Kardaras-Zitkovic} allow simultaneous variations on preferences and subjective probabilities. More recently, \cite{Mocha-Westray} focus on the power utility maximization problem and investigate stability respect to relative risk aversion, market price of risk, and investment constraints.

In \cite{Larsen} and \cite{Kardaras-Zitkovic}, convergence in probability of optimal payoffs is obtained under an uniform integrability assumption. One can prove that the optimal investment strategies also converge; cf. Remark \ref{rem: comparsion}. Our uniform bound on the ratio of marginal utilities implies an analogous integrability condition; cf. Remark \ref{rem: UI}. However the additional structure imposed here allows us to obtain more precise information on how fast the convergence takes place.

A different type of stability problem is studied in \cite{Larsen-Zitkovic}. Therein stability of the optimal payoff with respect to market variations is studied while a utility defined on $\Real_+$ is fixed. This type of stability problem has recently been investigated in \cite{Frei} and \cite{Bayraktar-Kravitz} for the exponential utility maximization problem.

In the second setting, we consider a sequence of utility random fields $(\mathcal{U}_p)_{p<0}$, each of which is of the form $\mathcal{U}_p = D \, U_p$ for a positive random variable $D$ and a utility function $U_p$ defined on $\Real_+$. For each $U_p$, the ratio of its marginal utility with respect to $x^{p-1}$ is bounded from above and away from zero. In this sense $U_p$ is comparable to  power utility $\wt{U}_p=x^p/p$ with constant relative risk aversion $1-p$. As the ratio of marginal utilities going to $1$, $(U_p)_{p<0}$ approaches $(\wt{U}_p)_{p<0}$ which converges to exponential utility, with appropriate domain shift, as $p\downarrow -\infty$  (cf. \cite[Remark 3.3]{Nutz-asy}).

Our second main result, Theorem \ref{thm: conv R_+}, states that, when the ratio of marginal utilities converges to $1$ at a rate at least as fast as the relative risk aversion going to infinity, then the optimal proportion invested in risky assets, scaled by $1-p$, converges to the optimal monetary value invested in risky assets in the exponential hedging problem. Therein $(1-p)^{-1}$ can be regarded as the rate of convergence. This result is first obtained in \cite{Nutz-asy} where $(U_p)_{p<0}$ is a sequence of power utilities. We complement Nutz's result by allowing deviation from power utility and analyze the impact on the convergence from the ratio of marginal utilities. On the dual side, the stability problem formulated  here is related to the convergence of optimal martingale measures which is studied in \cite{Grandits-Rheinlander}, \cite{Mania-Tevzadze}, and \cite{Santacroce}.

The starting point of our proofs in both settings is the following key result from the \emph{duality theory}: the optimal wealth process is a \emph{martingale} after multiplied by the optimal dual process, and a \emph{supermartingale} after multiplied by any other process in the dual domain. When random endowment presents, aforementioned  properties have been proved in \cite{Owen-Zitkovic} for utility defined on $\Real$ and in \cite{Karatzas-Zitkovic} for utility defined on $\Real_+$.
This property, combined with scaling properties of exponential (resp. power) utility, leads to an estimate on the difference (resp. ratio) of optimal payoffs for $U_\delta$ (resp. $U_p$) and exponential (resp. power) utility. The remaining proof does not depend on the market specifications. Therefore methods in this paper could potentially be applied to other market settings where the aforementioned property on the optimal wealth process holds, for example, markets with transaction cost, see \cite{Cvitanic-Karatzas}, and the utility maximization with forward criteria, see \cite{Musiela-Zariphopoulou}.

The structure of the paper is simple. After this introduction, Section \ref{sec: main results} describes the problems and states main results, while all proofs are given in Sections \ref{sec: stab R} and \ref{sec: stab R_+}.

\section{Main results}\label{sec: main results}

We consider a financial market of $d$-risky assets whose discounted prices are modeled by a locally bounded $\Real^d$-valued semimartingale $(S_t)_{t\in[0,T]}$, defined on a filtered probability space $(\Omega, \F, (\F_t)_{t\in [0,T]}, \prob)$, in which $\F_0$ coincides with the family of $\prob$-null sets and $(\F_t)_{t\in[0,T]}$ is right continuous. When price processes are non-locally bounded, we refer reader to \citep{Biagini-Frittelli-05, biagini-Frittelli-07}.

\subsection{Utilities defined on $\Real$}
Consider a sequence of standard utility functions\footnote{A standard utility function is  strictly increasing, strictly concave, and continuously differentiable.} $U_\delta:\Real \rightarrow \Real$, indexed by $\delta \geq 0$, converging in the following sense:

\begin{ass}\label{ass: utility conv}
 $\lim_{\delta\downarrow 0} U_{\delta}(x) = U_0(x)$ for $x\in \Real$, where $U_0(x) = -\exp(-x)$.\footnote{After appropriate scaling all results in this paper hold when $U_0$ has other value of absolute risk aversion.}
\end{ass}
The pointwise convergence of utility functions is widely used in the literature; e.g. \cite{Jouini-Napp} and \cite{Larsen}. The pointwise convergence, restricted to the class of concave functions (utility functions), implies a more economic meaningful mode of convergence: the pointwise (and hence locally uniformly) convergence of their derivatives (marginal utilities); see \cite[pp. 90 and pp. 248]{Rockafellar}. However the pointwise convergence is not enough for the stability of utility maximization problem; see an counterexample in \cite{Larsen}. We further restrict each $U_\delta$ to a class of utilities which are comparable to the exponential utility $-\frac{1}{\alpha_\delta} \exp(-\alpha_\delta x)$.

\begin{ass}\label{ass: marginal bound exp}
 There exist constants $0<\lR\leq 1 \leq \uR$ and $(\alpha_\delta)_{\delta>0}$ with $\lim_{\delta \downarrow 0} \alpha_\delta =1$ such that
 \[
 \lR \leq \fR_\delta(x) := \frac{U'_{\delta}(x)}{\exp(-\alpha_\delta x)} \leq \uR, \quad \text{ for all } \delta>0 \text{ and } x\in \Real.
 \]
\end{ass}

\begin{rem}\label{rem: utility bdd R}
This assumption implies that each $U_\delta$ is bounded from above. Indeed, integrating $\lR \exp(-\alpha_\delta x) \leq U'_\delta(x) \leq \uR \exp(-\alpha_\delta x)$ on $(0,\infty)$ yields $\lR/\alpha_\delta + U_\delta(0) \leq U_\delta(\infty) \leq \uR/\alpha_\delta + U_\delta(0)$\footnote{These bounds can be made uniform in $\delta$, since $\lim_{\delta\downarrow 0} U_\delta(0)=-1$ and $\lim_{\delta\downarrow 0}\alpha_\delta = 1$.}. Moreover, $U_\delta$ is sandwiched between two utilities with constant absolute risk aversion $\alpha_\delta$. To see this, integrating the previous bounds for $U'_\delta(x)$ on $(x,\infty)$ induces $U_\delta(\infty) - \frac{1}{\alpha_\delta} \uR \exp(-\alpha_\delta x) \leq U_\delta(x) \leq U_\delta(\infty) -\frac{1}{\alpha_\delta}\lR\exp(-\alpha_\delta x)$ for any $x\in \Real$.
One can also derive from Assumption \ref{ass: marginal bound exp} that each $U_\delta$ satisfies the Inada conditions, i.e., $\lim_{x\downarrow -\infty} U_{\delta}'(x) =\infty$ and $\lim_{x\uparrow \infty} U_{\delta}'(x) =0$, and $U_\delta$ has \emph{reasonable asymptotic elasticity}, i.e.,
\[
AE_{-\infty}(U_\delta):= \liminf_{x\downarrow -\infty} \frac{xU_\delta'(x)}{U_\delta(x)}>1 \quad \text{ and } \quad AE_{\infty}(U_\delta) := \limsup_{x\uparrow \infty} \frac{x U'_{\delta}(x)}{U(x)}<1.
\]
Hence each $U_\delta$ is reasonable risk averse at high and low wealth limit; cf. \citep{Kramkov-Schachermayer-99, Kramkov-Schachermayer-03}.
\end{rem}

To introduce the utility maximization problem considered, we denote by $M^a$ (resp. $M^e$) the class of probability measures $\wt{\prob} \ll \prob$ (resp. $\wt{\prob}\sim \prob$) such that $S$ is a local martingale under $\wt{\prob}$. Consider the convex conjugate $V_\delta: (0,\infty) \rightarrow \Real$ defined by $V_\delta(y):= \sup_{x\in \Real} (U_\delta(x) - x y)$. The \emph{generalized entropy} of $\wt{\prob}\in M^a$ relative to $\prob$ is defined as $\expec_\prob[V_\delta(d\wt{\prob}/d\prob)] \in (0,\infty]$. We denote by $\mathcal{M}^a_\delta$ (resp. $\mathcal{M}^e_\delta$) the set of probability measures $\wt{\prob}\in M^a$ (resp. $\wt{\prob}\in M^e$) with finite generalized entropy. Even though definition of $\mathcal{M}^a_\delta$ (resp. $\mathcal{M}^e_\delta$) depends on $V_\delta$, Lemma \ref{lem: M^a indep delta} below shows that all $\mathcal{M}^a_\delta$ (resp. $\mathcal{M}^e_\delta$) are the same for $\delta\geq 0$ under Assumption \ref{ass: marginal bound exp}. Henceforth we drop the subscript $\delta$ to write $\mathcal{M}^a$ (resp. $\mathcal{M}^e$) instead.

There is an agent whose preference is described by one of the utility function $U_\delta$. She is able to trade in the financial market and has a random endowment $\xi_\delta$ which is an $\F_T$-measurable random variable. Following \cite{Owen-Zitkovic}, we assume that $\xi_\delta$ is potentially unbounded but can be super-hedged.

\begin{ass}\label{ass: random endowment}
 There exist $x_\delta, \wt{x}_\delta \in \Real$ and a predictable $S$-integrable process $G_\delta$ such that
 \[
  x_\delta \leq \xi_\delta \leq \wt{x}_\delta + G_\delta \cdot S_T, \quad \text{ for each } \delta \geq 0,
 \]
 where $G_\delta \cdot S$ is $\prob$-a.s. uniformly bounded from below by a constant and $G_\delta \cdot S_T$ stands for $\int_0^T G_{\delta,t} dS_t$.
\end{ass}

When the utility function is defined on $\Real$, the class of wealth processes with uniform lower bound is not large enough for the problem considered below; cf. \cite{Schachermayer}. Therefore we recall the following class of permissible strategies from \cite{Owen-Zitkovic}: $H$ is a \emph{permissible trading strategy} if it is inside
\[
 \mathcal{H}^{perm}:= \left\{H\,:\, \begin{array}{l}H \text{ is a } \text{predictable}, S\text{-integrable process such that } \\
 H\cdot S \text{ is a } \wt{\prob}\text{-supermartingale for all } \wt{\prob}\in \mathcal{M}^a\end{array}\right\}.\footnote{Since $\mathcal{M}^a$ is the same for different $\delta$, $\mathcal{H}^{perm}$ is independent of $\delta$ as well. Therefore even though the utility of the agent may change with respect to $\delta$, she always choose trading strategy from the same permissible class.}
\]

Our agent chooses permissible strategies to maximize her utility on wealth and endowment at an investment horizon $T$:
\begin{equation}\label{def: u_delta}
 u_\delta:= \sup_{H\in \mathcal{H}^{perm}}\expec_\prob\bra{U_\delta \pare{H\cdot S_T + \xi_\delta}}.
\end{equation}
In order to ensure the existence of the optimal strategy, we impose
\begin{ass}\label{ass: finite entropy}
 $\mathcal{M}^e \neq \emptyset$.
\end{ass}
When $U_\delta$ has reasonable asymptotic elasticity, $\mathcal{M}^a \neq \emptyset$, and Assumption \ref{ass: random endowment} holds, Assumption \ref{ass: finite entropy} is actually the necessary and sufficient condition for the existence of optimal strategy for \eqref{def: u_delta}; cf. \cite[Theorem 1.9]{Owen-Zitkovic}. We further recall the following result from \cite{Owen-Zitkovic}.

\begin{prop}[Owen-\v{Z}itkovi\'{c}]\label{prop: u_delta wellposed}
 Let $U_\delta$ be of reasonable asymptotic elastic and Assumptions \ref{ass: random endowment} and \ref{ass: finite entropy} hold. Then there exists an optimal strategy $H_\delta \in \mathcal{H}^{perm}$ for \eqref{def: u_delta} such that $H_\delta \cdot S$ is a $\wt{\prob}$-supermartingale for all $\wt{\prob}\in \mathcal{M}^a$ and a $\qprob_\delta$-martingale for some $\qprob_\delta \in \mathcal{M}^e$, whose density $d\qprob_\delta/d\prob$ satisfies
 \[
  y_\delta \frac{d\qprob_\delta}{d\prob} = U'_\delta \pare{H_\delta\cdot S_T + \xi_\delta}, \quad \text{ for some positive constant } y_\delta.
 \]
\end{prop}

In the previous result, $\qprob_0$ is the the \emph{minimal entropy measure} which minimizes $\expec_\prob[V_0(d\wt{\prob}/d\prob)]$, with $V_0(y)=y\log y -y$, among all $\wt{\prob}\in M^a$. To simplify notation we drop the subscript $0$ and denote the minimal entropy measure by $\qprob$.
In order to investigate the convergence of \eqref{def: u_delta} and its optimal strategy as $\delta\downarrow 0$. We assume the following convergence of random endowments.

\begin{ass}\label{ass: endowment conv}
 There exists a constant $C\in \Real_+$ such that $\alpha_\delta \xi_\delta - \xi_0 \geq -C$, $\prob$-a.s. for all $\delta >0$. Moreover $\lim_{\delta \downarrow 0} \expec_\qprob[|\alpha_\delta \xi_\delta - \xi_0|]=0$.
\end{ass}

The previous assumption clearly holds when $(\xi_\delta)_{\delta \geq 0}$ is uniformly bounded and $\qprob-\lim_{\delta\downarrow 0} \xi_\delta =\xi_0$, where $\qprob-\lim$ represents convergence in probability $\qprob$. Denote the optimal payoff by $X^\delta_T = H_\delta \cdot S_T$ for $\delta \geq 0$. The first main result states the convergence of $X^\delta_T$, its associated strategy, and $u_\delta$, as $\delta \downarrow 0$.

\begin{thm}\label{thm: conv R}
 Let Assumptions \ref{ass: utility conv}, \ref{ass: marginal bound exp}, \ref{ass: random endowment}, \ref{ass: finite entropy} and \ref{ass: endowment conv} hold. Then the following statements hold:
 \begin{enumerate}
  \item[i)] $\lim_{\delta \downarrow 0} \expec_\qprob[|X^\delta_T - X^0_T|] =0$;
  \item[ii)] $\lim_{\delta \downarrow 0} u_\delta = u_0$;
  \item[iii)] If $S$ is continuous then
  \[
  \lim_{\delta \downarrow 0} \expec_\qprob\bra{\pare{\int_0^T (H_\delta-H_0)^{\top}_t d\langle S\rangle_t (H_\delta -H_0)_t}^{p/2}}=0, \quad \text{ for any } p\in (0,1).
  \]
 \end{enumerate}
\end{thm}

\begin{rem}\label{rem: comparsion}
 When $(U_\delta)_{\delta\geq 0}$ are defined on $\Real_+$, the analogue result has been proved in \cite{Larsen} and \cite{Kardaras-Zitkovic}. Therein $\prob-\lim_{\delta \downarrow 0} X^\delta_T / X^0_T =1$ and $\lim_{\delta\downarrow 0} u_\delta = u_0$ are proved. Define $\overline{\prob}$ via $d\overline{\prob}/d\prob = c U_0'(X^0_T) X^0_T$ for a normalization constant $c$. Then $X^0$ has the num\'{e}raire property under $\overline{\prob}$, i.e., $X^\delta/X^0$ is a $\overline{\prob}$-supermartingale. Then $\lim_{\delta\downarrow 0}\expec_{\overline{\prob}}[|X^\delta_T / X^0_T -1|]=0$ and the convergence of the associated strategies follow from \cite[Theorem 2.5]{Kardaras}.
\end{rem}

\begin{rem}\label{rem: UI}
 In \cite{Larsen} and \cite{Kardaras-Zitkovic}, an uniform integrability assumption is the key to stability. Assumption \ref{ass: marginal bound exp} implies an analogue condition is satisfied. Indeed,
 Remark \ref{rem: utility bdd R} implies that $(U_\delta)_{\delta\geq 0}$ is uniformly bounded from above by
 \[
  U_* (x) := \frac{\uR}{\alpha_*} + \frac{\lR}{\alpha_*} \pare{1-\exp(-\alpha_* x)}, \quad \text{ where } \alpha_* = \min_{\delta\geq 0} \alpha_\delta.
 \]
 Since $\expec_{\prob}[V_*(d\qprob/d\prob)]<\infty$, where $V_*$ is the convex conjugate of $U_*$ and dominates all $V_\delta$, $\{V_\delta(d\qprob/d\prob)\}_{\delta\geq 0}$ is then clearly uniformly integrable under $\prob$. However the additional structure in Assumption \ref{ass: marginal bound exp} allows us to discuss the rate of convergence in what follows.
\end{rem}

Let us describe rates of convergence for the ratio of marginal utilities, absolute risk aversion, and random endowments via
\[
 f(\delta) := \sup_{x\in \Real} \left|\fR_\delta(x) -1\right|, \quad g(\delta) := |\alpha_\delta -1|, \quad \text{ and } \quad h(\delta):= \expec_\qprob\bra{|\xi_\delta -\xi_0|^2}, \quad \text{ for } \delta \geq 0.
\]

\begin{cor}\label{cor: convergence rate}
 Let Assumptions \ref{ass: utility conv}, \ref{ass: marginal bound exp}, and \ref{ass: finite entropy} hold. Suppose that $\xi_\delta$ is bounded uniform in $\delta$, moreover
 $\lim_{\delta \downarrow 0} f(\delta) =\lim_{\delta \downarrow 0} g(\delta) = \lim_{\delta \downarrow 0} h(\delta) = 0$. Then
 \[
  \expec_\qprob\bra{|X^\delta_T - X^0_T|} \sim O\pare{f(\delta)^2 + g(\delta)+ h(\delta)}, \quad \text{ for sufficiently small } \delta.
 \]
\end{cor}

\begin{rem}\label{rem: conv rate}
 When $U_\delta$ is the exponential utility with risk aversion $a_\delta$ and no random endowment presnets, it is clear that $X^\delta_T = X^0_T/\alpha_\delta$ converges to $X^0_T$ at the rate of $g(\delta)$. When $U_\delta$ deviates from exponential utility and random endowment presents, the rate of convergence for the optimal payoff is determined by three components: convergence of the ratio of marginal utilities, convergence of absolute risk aversions, and convergence of random endowments. Corollary \ref{cor: convergence rate} shows that the rate of convergence is at least second order on the first component, first order on the second and third components. This provides a quantitative measure on how far $X^\delta_T$ is away from $X^0_T$.

 The convergence rate for optimal strategies can also be determined.
 When $S$ is continuous, Corollary \ref{cor: convergence rate} and Burkholder-Davis-Gundy inequality combined imply
 \[
  \expec_\qprob\bra{\pare{\int_0^T (H_\delta-H_0)^{\top}_t d\langle S\rangle_t (H_\delta -H_0)_t}^{p/2}} \sim O \pare{f(\delta)^2 + g(\delta) + h(\delta)},
 \]
 for any $p\in (0,1)$ and small $\delta$;
 see Lemma \ref{lem: Sp norm local mart} and Corollary \ref{cor: quad-var delta X} for more details. Here $H_0$ is the hedging strategy in the exponential hedging problem (cf. \cite{six-authors}, \cite{Kabanov-Stricker}).
\end{rem}

Another application of Theorem \ref{thm: conv R} is the stability of utility-based prices with respect to agent's preference.
Consider a contingent claim $B\in L^\infty(\F_T)$. An agent, endowed with utility $U_\delta$ and endowment $\xi_\delta$, takes her preference into account to price the claim $B$ as
\[
 \expec_{\qprob_\delta}[B],
\]
where $\qprob_\delta$ is introduced in Proposition \ref{prop: u_delta wellposed}. This price is called \emph{fair price} (Davis price), cf. \cite{Davis}.
Theorem \ref{thm: conv R} implies the continuity of fair price with respect to agent's preference.
\begin{cor}\label{cor: davis price}
 Let Assumptions \ref{ass: utility conv}, \ref{ass: marginal bound exp}, \ref{ass: random endowment}, \ref{ass: finite entropy} and \ref{ass: endowment conv} hold. Then
 \[
  \lim_{\delta\downarrow 0} \expec_{\qprob_\delta}[B]= \expec_\qprob[B].
 \]
\end{cor}

Another utility-based pricing is the \emph{indifference price} introduced into mathematical finance by \cite{Hodges-Neuberger}; See \cite{Carmona} and references therein for recent development on this topic.
Given an agent endowed with utility $U_\delta$ and an initial wealth $x_0\in \Real$, her \emph{indifference buyer's price}, $p_\delta=p(B, x, U_\delta)$, of $B$ is defined as the solution to the equation
\[
 u_\delta(x_0+B- p_\delta) = u_\delta(x_0),
\]
where $u_\delta(\zeta)$ is defined in \eqref{def: u_delta} with $\xi_\delta = \zeta$. The existence and uniqueness of $p_\delta$ is proved in \cite[Proposition 7.2]{Owen-Zitkovic}. Theorem \ref{thm: conv R} ii) allows us to establish the following stability property of the indifference buyer's price with respect to agent's preference.

\begin{cor}\label{cor: indifference price}
 Let Assumptions \ref{ass: utility conv}, \ref{ass: marginal bound exp}, and \ref{ass: finite entropy} hold.
 Then $\lim_{\delta \downarrow 0} p_\delta = p_0$.
\end{cor}

\begin{rem}
 The continuity of Davis prices and indifference prices with respect to agent's preference has been investigated in \cite{Carassus-Rasonyi} in a discrete time market with bounded stock price processes.
\end{rem}

\subsection{Utilities defined on $\Real_+$}
We continue with our second main result, which concerns the convergence of problems with utilities defined on $\Real_+$ to the exponential utility maximization problem. Consider a sequence of utility random fields $\mathcal{U}_p: \Omega \times \Real_+ \rightarrow \Real$, indexed by $p<0$, each of which is of the form \[
 \mathcal{U}_p(x) = D \,U_p(x), \quad  x\in \Real_+,
\]
where $D$ is a $\F_T$-measurable positive random variable and $U_p:\Real_+ \rightarrow \Real$ is a standard utility function. We assume that each $U_p$ is comparable to power utility $x^p/p$ in the following sense:

\begin{ass}\label{ass: marginal bound power}
 There exist constants $0<\lR_p\leq 1\leq \uR_p$ such that
 \[
  \lR_p \leq \fR_p(x):=\frac{U'_p(x)}{x^{p-1}} \leq \uR_p, \quad \text{ for all } x\in \Real_+.
 \]
\end{ass}
\begin{rem}\label{rem: utility bdd R_+}
The previous assumption implies that each $U_p$ is bounded from above. Indeed, integrating $U'_p(x) \leq \uR_p \,x^{p-1}$ on $(1, x)$ yields $U_p(x)\leq U_p(1) + \uR_p (x^p/p  -1/p)\leq U_p(1) -\uR_p/p$ for $x\geq 1$ and $p<0$. Moreover $U_p$ is sandwiched between two utilities with relative risk aversion $1-p$. To see this, integrating $\lR_p x^{p-1}\leq U'_p(x) \leq \uR_p x^{p-1}$ on $(1, x)$ when $x\geq 1$ or on $(x, 1)$ when $x<1$ yields $c_p(x^p/p-1/p) + U_p(1) \leq U_p(x) \leq C_p(x^p/p-1/p) + U_p(1)$, for $x>0$ and some constants $c_p$ and $C_p$. Furthermore each $U_p$ satisfies the Inada condition, i.e., $\lim_{x\downarrow 0} U'_p(x) =\infty$ and $\lim_{x\uparrow \infty} U'_p(x) =0$, and $U_\delta$ has reasonable asymptotic elasticity, i.e., $AE_\infty(U_\delta)<1$.
\end{rem}

The discounted prices of risky assets are specified to be stochastic exponential $S= (\mathcal{E}(R^1), \cdots, \mathcal{E}(R^d))$, where $R$ is an $\Real^d$-valued \cadlag\, locally bounded semimartingale with $R_0=0$.
The agent is endowed with the utility random field $\mathcal{U}_p$ and an initial capital $x_0\in \Real_+$. A trading strategy is a predictable $R$-integrable $\Real^d$-valued process $\pi$ whose $i$-th component $\pi^i$ represents the fraction of current wealth invested in the $i$-th risky asset. Then the associated wealth process $X(\pi)$ satisfies
\[
 X_t = x_0 + \int_0^t X_{s-} \pi_s \,dR_s, \quad 0\leq t\leq T.
\]
A trading strategy is \emph{admissible} if the associated wealth process is strictly positive. We denote by $\mathcal{A}(x_0)$ the class of admissible trading strategies. For an admissible strategy $\pi$, $H^i:= \pi^i X / S^i_- \indic_{\{S^i_- \neq 0\}}$ corresponds to the number of shares invested in the $i$-th asset.

The agent chooses admissible trading strategies to maximize her utility of payoff:
\begin{equation}\label{def: u_p}
 u_p(x_0):= \sup_{\pi\in \mathcal{A}(x_0)} \expec_\prob\bra{D U_p(X_T(\pi))}.
\end{equation}
The dependence of $u_p$ on $x_0$ will be omitted if no confusion is caused.
Since $U_p$ is bounded from above, $u_p(x_0)<\infty$ whenever $D_T$ has finite $\prob$-expectation. We recall the following version of Theorem 3.10 from \cite{Karatzas-Zitkovic}.

\begin{prop}[Karatzas-\v{Z}itkovi\'{c}]\label{prop: u_p wellposed} Assume that the set of equivalent local martingale measures for $S$ is not empty, moreover there exist constants $0<k_1\leq k_2<\infty$ such that $k_1\leq D \leq k_2$. Then for each $p<0$ there exists an optimal strategy $\pi_p\in \mathcal{A}(x_0)$ for \eqref{def: u_p}. The associated wealth process $X^{(p)}$ satisfies
\[
 y_p Y^{(p)}_T = D U'_p(X^{(p)}_T),
\]
where $y_p= u'_p(x_0)$ and $Y^{(p)}$ is some supermartingale deflator with $Y^{(p)}_0=1$. Moreover
\[
 y_p \,x_0 = \expec_\prob\bra{D U'_p(X^{(p)}_T) X^{(p)}_T} \geq \expec_\prob\bra{D U'_p (X^{(p)}_T) X_T},
\]
for any admissible wealth process $X$.
\end{prop}

To state our second main result, let us recall the exponential hedging problem. Given a contingent claim $B\in L^\infty(\F_T)$, the agent choose permissible strategy to maximize the expected exponential utility of the terminal wealth including the claim,
\begin{equation}\label{def: exp hedge}
 \sup_{\vartheta \text{ permissible}} \expec_\prob\bra{-\exp(B-x_0 - \vartheta\cdot R_T)}.
\end{equation}
Here $\vartheta$ is the monetary value invested in the risky assets. Its corresponding number of shares is $H^i:= \vartheta^i/S^i_- \indic_{\{S^i_-\neq 0\}}$ which satisfies $H\cdot S = \vartheta \cdot R$. The strategy $\vartheta$ is permissible if its corresponding $H\in \mathcal{H}^{perm}$. When $S$ is locally bounded, \eqref{def: exp hedge} admits an optimal strategy $\hat{\vartheta}$; cf. \cite[Theorem 2.1]{Kabanov-Stricker}.

We impose the following assumption on filtration which is satisfied for the Brownian filtration.

\begin{ass}\label{ass: cont filtration}
 The filtration $(\F_t)_{t\in[0,T]}$ is continuous, i.e., all $\F$-local martingales are continuous.
\end{ass}
The previous assumption implies that $S$ is continuous. Hence $R$ satisfies the \emph{structure condition}:
\[
 R=M+ \int d\langle M\rangle \lambda,
\]
where $M$ is a continuous local martingale with $M_0=0$ and $\lambda\in L^2_{loc}(M)$; cf. \cite{Schweizer}.

Our second main result studies the asymptotic behavior of the optimal strategy $\pi_p$  for \eqref{def: u_delta} as $p\downarrow -\infty$.

\begin{thm}\label{thm: conv R_+}
 Let Assumptions \ref{ass: finite entropy}, \ref{ass: marginal bound power}, and \ref{ass: cont filtration} hold. Set $D= \exp(B)$ for $B\in L^\infty(\F_T)$. If  $\lR_p$ and $\uR_p$ in Assumption \ref{ass: marginal bound power} satisfy
 \begin{equation}\label{ass: Rp conv}
  \limsup_{p\downarrow -\infty} \,(1-p) \,(\uR_p-1) <\infty \quad \text{ and } \quad \limsup_{p\downarrow -\infty} \,(1-p)\, (1-\lR_p) <\infty,
 \end{equation}
 then
 \[
  \prob-\lim_{p\downarrow -\infty} \int_0^T \pare{(1-p) \pi_p - \hat{\vartheta}}^{\top}_t d\langle M \rangle_t \pare{(1-p) \pi_p - \hat{\vartheta}}_t =0.
 \]
\end{thm}

This result states that whenever the ratio of marginal utilities converges to $1$ at least as fast as the relative risk aversion converging to infinity, the optimal fraction invested in risky assets in the power type problem, after scaled by $1-p$, converges to the optimal monetary value invested in the exponential hedging problem. Here $(1-p)^{-1}$ can be considered as the rate of convergence.

\begin{rem}\label{rem: U_p construction}
 Given a utility function $U$ such that
 \[
  \lR \leq \frac{U'(x)}{x^{p_0-1}} \leq \uR, \quad \text{ for all } x>0,
 \]
 where $0<\lR \leq 1\leq \uR$ and $p_0<0$, there exists a family of utilities $(U_p)_{p\leq p_0}$ such that $U_{p_0} = U$ and \eqref{ass: Rp conv} is satisfied for some sequences $(\uR_p)_{p\leq p_0}$ and $(\lR_p)_{p\leq p_0}$. Indeed, take any function $f:(-\infty, 0)\rightarrow (0,1)$ such that $f(p_0)=1$ and $\limsup_{p\downarrow -\infty} \,(1-p) \, f(p) <\infty$. Set
 \[
  U'_p(x) = f(p) x^{p-p_0} U'(x) + (1-f(p)) x^{p-1}, \quad \text{ for } p\leq p_0.
 \]
 One can check that $U_p$ is a standard utility function and
 \[
  \lR_p := f(p) (\lR-1)+ 1 \leq \frac{U'_p(x)}{x^{p-1}} \leq f(p) (\uR-1)+ 1 =: \uR_p,
 \]
 where both $\limsup_{p\downarrow -\infty} \,(1-p) \, (1-\lR_p)$ and $\limsup_{p \downarrow -\infty} (1-p) \, (\uR_p-1)$ are finite.
\end{rem}

\begin{rem}
 Denote by $\wt{\pi}_p$ the optimal strategy for \eqref{def: u_p} when $U_p= x^p/p$. Nutz proved a remarkable result in \cite[Theorem 3.2]{Nutz-asy} that $(1-p)\wt{\pi}_p \rightarrow \hat{\vartheta}$ in $L^2_{loc}(M)$; cf. \cite[Lemma A.3]{Nutz-asy} for characterization of this convergence. In particular the previous convergence implies
 \begin{equation}\label{eq: (1-p)tpi-var}
  \prob-\lim_{p\downarrow -\infty} \int_0^T ((1-p) \wt{\pi}_p - \hat{\vartheta})^\top_t \,d\langle M\rangle_t \,((1-p) \wt{\pi}_p - \hat{\vartheta})_t =0.
 \end{equation}
 Therefore $\wt{\pi}_p$ converges to $\hat{\vartheta}$ at the rate of $(1-p)^{-1}$. We complement Nutz's result by showing that $\pi_p-\wt{\pi}_p$ converges to $0$ at the rate $(1-p)^{-1}$, when the ratio of marginal utilities converges to $1$ at least at the same rate. In particular, we prove
 \begin{equation}\label{eq: (1-p)tpi-pi}
  \prob-\lim_{p\downarrow -\infty} \int_0^T (1-p) (\wt{\pi}_p - \pi_p)^\top_t \,d\langle M\rangle_t \,(1-p) (\wt{\pi}_p - \pi_p)_t =0.
 \end{equation}
 Then Theorem \ref{thm: conv R_+} follows from combining the previous two convergence.
\end{rem}

\begin{rem}\label{rem: filtration}
 One can assume that both $S$ and the opportunity processes $(L^{(p)})_{p<0}$, recalled in Section \ref{sec: stab R_+}, are continuous instead of Assumption \ref{ass: cont filtration}, which is the most important and easy to check sufficient condition for the continuity of $S$ and $(L^{(p)})_{p<0}$. Only the continuity of $S$ is used to prove \eqref{eq: (1-p)tpi-pi}, continuity of both $S$ and $L^{(p)}$ for all $p<0$ are needed for \eqref{eq: (1-p)tpi-var}.
\end{rem}
\section{Stability for utilities defined on $\Real$}\label{sec: stab R}

Theorem \ref{thm: conv R} and its corollaries will be proved in this section. Let us start with the following property on the family $(\mathcal{M}^a_\delta)_{\delta \geq 0}$.

\begin{lem}\label{lem: M^a indep delta}
 Under Assumption \ref{ass: marginal bound exp}, all $\mathcal{M}^a_\delta$ (resp. $\mathcal{M}^e_\delta$) are the same for $\delta\geq 0$.
\end{lem}
\begin{proof}
 Denote $\wt{U}_\delta(x)= -\frac{1}{\alpha_\delta} \exp(-\alpha_\delta x)$ and $\wt{V}_\delta(y) = \frac{1}{\alpha_\delta} y \log{y} - \frac{y}{\alpha_\delta}$ to be its convex conjugate. Here $\alpha_\delta$ converges to $a_0:=1$ as $\delta \downarrow 0$. Set $y= U'_\delta(x)$, which can take arbitrary value in $(0,\infty)$ as $x$ varies in $\Real$. It follows from Assumption \ref{ass: marginal bound exp} that $y/\uR \leq \wt{U}'_\delta(-V'_\delta (y)) \leq y/\lR$, which implies $\wt{V}'_\delta(y/\uR) \leq V'_\delta (y) \leq \wt{V}'_\delta(y/\lR)$ for any $y\in (0,\infty)$. Integrating the previous inequalities on $(0, y)$ and utilizing $\wt{V}_\delta(0)= \wt{U}_\delta(\infty) =0$, we obtain
 \begin{equation*}
  \uR \,\wt{V}_\delta(y/\uR) + V_\delta(0) \leq V_\delta(y) \leq \lR \wt{V}_\delta(y/\lR) + V_\delta(0).
 \end{equation*}
 Recall from Remark \ref{rem: utility bdd R} that $(U_\delta(\infty))_{\delta>0}$ is uniformly bounded. Then there exists $N$ such that $-N\leq V_\delta(0)= U_\delta(\infty)\leq N$ for any $\delta$. The previous two inequalities combined yield
 \[
  \frac{1}{\alpha_\delta} \wt{V}_0(y) - \frac{1}{\alpha_\delta} y \log \uR -N \leq V_\delta(y) \leq \frac{1}{\alpha_\delta} \wt{V}_0(y) - \frac{1}{\alpha_\delta} y \log \lR +N, \quad \text{ for any } y.
 \]
 Therefore $\expec_\prob[V_\delta(d\wt{\prob}/d\prob)]<\infty$ if and only if $\expec_\prob[\wt{V}_0(d\wt{\prob}/d\prob)]<\infty$.
\end{proof}

To prove Theorem \ref{thm: conv R}, observe that, without loss of generality all $(\alpha_\delta)_{\delta \geq 0}$ in Assumption \ref{ass: marginal bound exp} can be assumed to be $1$. Indeed, define $\overline{U}_\delta(x) := \alpha_\delta U_\delta(x/\alpha_\delta)$.  Assumption \ref{ass: marginal bound exp} implies
\[
 \lR \leq \frac{\overline{U}'_\delta(x)}{\exp(-x)} \leq \uR, \quad \text{ for any } x\in \Real.
\]
Moreover, $\overline{U}(x)$ converges to $-\exp(-x)$ pointwise, since $\alpha_\delta$ converges to $1$ and $U_\delta(x)$ converges to $-\exp(-x)$ locally uniformly; see \cite[pp. 90]{Rockafellar}. Therefore \eqref{def: u_delta} can be rewritten as
\[
 u_\delta = \frac{1}{\alpha_\delta} \sup_{H\in \mathcal{H}^{perm}}\expec_\prob\bra{\overline{U}_\delta\pare{\alpha_\delta H\cdot S_T+ \alpha_\delta \xi_\delta}} = \frac{1}{\alpha_\delta} \sup_{\overline{H}\in \mathcal{H}^{perm}}\expec_\prob\bra{\overline{U}_\delta\pare{\overline{H}\cdot S_T+ \overline{\xi}_\delta}},
\]
where $\overline{\xi}_\delta:= \alpha_\delta \xi_\delta$. Therefore the optimal strategy $H_\delta$ for \eqref{def: u_delta} is exactly $\overline{H}_\delta/\alpha_\delta$ where $\overline{H}_\delta$ maximizes the rightmost problem. Hence we can consider \eqref{def: u_delta} with utility $\overline{U}_\delta$ and the random endowment $\overline{\xi}_\delta$. In this case Assumption \ref{ass: marginal bound exp} holds with $\alpha_\delta =1$ for all $\delta\geq 0$.

Now suppose that Theorem \ref{thm: conv R} holds for $\overline{U}_\delta$, then the same statements hold for $U_\delta$ as well. For example, if $\lim_{\delta\downarrow 0} \expec_\qprob\bra{\left|(\overline{H}_\delta - \overline{H}_0)\cdot S_T\right|} =0$, then
\begin{equation}\label{eq: exp diff delta X}
\begin{split}
 \expec_\qprob\bra{\left|(H_\delta - H_0)\cdot S_T\right|} &= \frac{1}{\alpha_\delta} \expec_\qprob\bra{\left|(\overline{H}_\delta - \alpha_\delta H_0)\cdot S_T\right|}\\
 &\leq \frac{1}{\alpha_\delta} \expec_\qprob\bra{\left|(\overline{H}_\delta - \overline{H}_0)\cdot S_T\right|} + \frac{|\alpha_\delta -1|}{\alpha_\delta} \expec_\qprob\bra{\left|H_0\cdot S_T\right|}\\
 &\rightarrow 0, \quad \text{ as } \delta \downarrow 0,
\end{split}
\end{equation}
where $\overline{H}_0 = H_0$ and $\expec_\qprob[|H_0\cdot S_T|]<\infty$ since $H_0 \cdot S$ is a $\qprob$-martingale. Therefore, due to the previous change of variable, it suffices to prove Theorem \ref{thm: conv R} when
\begin{equation}\label{eq: a_delta=1}
\alpha_\delta =1, \quad  \text{ for all } \delta\geq 0.
\end{equation}
To this end, Theorem \ref{thm: conv R} i) will be proved in Corollary \ref{cor: L^1 conv diff X}, ii) in Proposition \ref{prop: value fun conv}, and iii) in Corollary \ref{cor: quad-var delta X}.
In the rest of this section, Assumptions \ref{ass: utility conv}, \ref{ass: marginal bound exp}, \ref{ass: random endowment}, \ref{ass: finite entropy} and \ref{ass: endowment conv} are enforced. To simplify notation, we introduce
\[
 X^\delta := H_\delta \cdot S, \quad \cX^\delta := X^\delta + \xi_\delta, \quad \Delta \xi_\delta := \xi_\delta - \xi_0, \quad \text{ and } \quad \Delta \cX^\delta := \cX^\delta - \cX^0, \quad \text{ for } \delta \geq 0.
\]

Proof of Theorem \ref{thm: conv R} i) starts with the following estimate.

\begin{lem}\label{lem: exp est R}
 It holds that
 \[
  \lim_{\delta \downarrow 0} \expec_\qprob\bra{\left|1- \fR_\delta(\cX^\delta_T) \exp(-\Delta \cX^\delta_T)\right| \left|\Delta X^\delta_T\right|}=0.
 \]
\end{lem}

\begin{proof}
 Recall from Proposition \ref{prop: u_delta wellposed} that $X^0$ is a $\qprob_\delta$-supermartingale and $X^\delta$ is a $\qprob_\delta$-martingale, where the density $d\qprob_\delta/d\prob$ is $U'_\delta(\cX^\delta_T)$ up to a constant. Therefore $U'_\delta(\cX^\delta_T) X^0_\cdot$ is a $\prob$-supermartingale and $U'_\delta(\cX^\delta_T) X^\delta_\cdot$ is a $\prob$-martingale. Since both these two processes have initial value zero, therefore $\expec_\prob\bra{U'_\delta(\cX^\delta_T) X^0_T} \leq 0 = \expec_\prob\bra{U'_\delta(\cX^\delta_T) X^\delta_T}$, which induces
 \[
  \expec_\prob\bra{U'_\delta (\cX^\delta_T) (X^0_T - X^\delta_T)}\leq 0.
 \]
 Similarly, the previous argument applied to $\qprob$ gives
 \[
  \expec_\prob\bra{U'_0(\cX^0_T)(X^\delta_T - X^0_T)}\leq 0.
 \]
 Summing up the previous two inequalities and changing to the measure $\qprob$ whose density is $U'_0(\cX^0_T)$  up to a constant, we obtain
 \[
  \expec_\qprob\bra{\pare{1-\frac{U'_\delta(\cX^\delta_T)}{U'_0(\cX^0_T)}}\pare{X^\delta_T - X^0_T}} \leq 0.
 \]
 Observe that the random variable in the expectation of the previous inequality is negative only when $X^0_T \geq X^\delta_T \geq I_\delta \pare{U'_0(\cX^0_T)} - \xi_\delta$ or $I_\delta\pare{U'_0(\cX^0_T)} -\xi_\delta \geq X^\delta_T \geq X^0_T$, where $I_\delta = (U'_\delta)^{-1}$. In either cases,
 \[
  \pare{\pare{1-\frac{U'_\delta(\cX^\delta_T)}{U'_0(\cX^0_T)}}\pare{X^\delta_T - X^0_T}}_- \leq \pare{\frac{U'_\delta(X^0_T + \xi_\delta)}{U'_0(X^0_T + \xi_0)}-1}\pare{I_\delta\pare{U'_0(\cX^0_T)}-\xi_\delta - X^0_T},
 \]
 where $(\cdot)_-$ represents the negative part.
 Utilizing the fact that $\expec_\qprob[|A|] \leq 2\expec_\qprob[A_-]$ for any random variable $A$ with $\expec_\qprob[A]\leq 0$, we obtain
 \[
  \expec_\qprob\bra{\left|\pare{1-\frac{U'_\delta(\cX^\delta_T)}{U'_0(\cX^0_T)}}\pare{X^\delta_T - X^0_T}\right|} \leq 2 \,\expec_\qprob\bra{\pare{\frac{U'_\delta(X^0_T + \xi_\delta)}{U'_0(X^0_T + \xi_0)} -1}\pare{I_\delta\pare{U'_0(\cX^0_T)}-\xi_\delta - X^0_T}}.
 \]
 Note that the left side of the previous inequality is $\expec_\qprob\bra{\left|1- \fR_\delta(\cX^\delta_T) \exp(-\Delta\cX^\delta_T)\right||\Delta X^\delta_T|} $. The statement follows once the expectation on the right side converges to zero as $\delta \downarrow 0$.

 To prove the desired convergence, let us first estimate the upper bound of $|I_\delta(U'_0(x)) - x|$ on $\Real$. Set $y= U'_0(x)$. It follows
 \begin{align*}
  I_\delta\pare{U'_0(x)} - x & = I_\delta (y) - I_0(y) = -\log \bra{\exp\pare{-\pare{I_\delta(y) - I_0(y)}}}\\
  &= -\log \bra{\frac{\exp(-I_\delta(y))}{y}} = \log \bra{\frac{U'_\delta(I_\delta(y))}{U'_0(I_\delta(y))}}.
 \end{align*}
 Assumption \ref{ass: marginal bound exp} then implies
 \[
  \sup_{x\in \Real} \left|I_\delta(U'_0(x)) -  x\right| \leq \max\{\log \uR, \log 1/\lR\} =: \eta.
 \]
 As a result, $\left|I_\delta(U'_0(\cX^0_T)) - X^0_T - \xi_\delta\right| \leq \eta  + |\Delta \xi_\delta|$. Assumptions \ref{ass: marginal bound exp} and \ref{ass: endowment conv} combined imply that
 \[
  \frac{U'_\delta(X^0_T + \xi_\delta)}{U'_0(X^0_T + \xi_0)} = \fR_\delta(X^0_T+\xi_\delta) \exp(-\Delta \xi_\delta) \leq \uR e^C.
 \]
 The previous two estimates combined yield
 \begin{equation}\label{eq: ub est}
  \left|\frac{U'_\delta(X^0_T + \xi_\delta)}{U'_0(X^0_T + \xi_0)} -1\right|\left|I_\delta\pare{U'_0(\cX^0_T)-\xi_\delta - X^0_T}\right| \leq (\uR e^C+1)(\eta + |\Delta \xi_\delta|),
 \end{equation}
 where the right side is uniformly integrable in $\delta$ under $\qprob$ thanks to $\lim_{\delta \downarrow 0}\expec_\qprob[|\Delta \xi_\delta|]=0$ in Assumption \ref{ass: endowment conv}. On the other hand, the term on the left side of \eqref{eq: ub est} converges to $0$ in probability $\qprob$. This follows from facts that $\limsup_{\delta\downarrow 0}|I_\delta(U'_0(\cX^0_T)) - \xi_\delta - X^0_T|$ is bounded and $\qprob-\lim_{\delta \downarrow 0} \fR_\delta(X^0_T + \xi_\delta) \exp(-\Delta \xi_\delta)=1$. The previous convergence follows from
 \[
 \begin{split}
 &\qprob\pare{|\fR_\delta(X^0_T + \xi_\delta) \exp(-\Delta \xi_\delta) -1|\geq \epsilon}\\
 \leq & \qprob\pare{|\fR_\delta(X^0_T + \xi_\delta) \exp(-\Delta \xi_\delta) -1|\geq \epsilon, |\xi_\delta|\leq N, |X^0_T|\leq N} + \qprob(|\xi_\delta| >N) +  \qprob(|X^0_T|>N),
 \end{split}
 \]
 where the first term on the right converges to $0$ as $\delta\downarrow 0$ since $\fR_\delta$ converges to $1$ locally uniformly and $\qprob-lim_{\delta\downarrow 0} \Delta \xi_\delta =0$, both second and third terms can be made arbitrarily small for sufficiently large $N$.
 The uniform integrability and convergence in probability combined imply
 \[
  \lim_{\delta\downarrow 0}\expec_\qprob\bra{\left|\frac{U'_\delta(X^0_T + \xi_\delta)}{U'_0(X^0_T + \xi_0)} -1\right|\left|I_\delta\pare{U'_0(\cX^0_T)-\xi_\delta - X^0_T}\right|}=0,
 \]
 hence the statement.
\end{proof}

The previous result provides a handle to study the $\mathbb{L}^1(\qprob)$ convergence of $X^\delta_T - X^0_T$.

\begin{cor}\label{cor: L^1 conv diff X}
 It holds that
 \[
  \lim_{\delta \downarrow 0} \expec_\qprob\bra{\left|\Delta X^\delta_T\right|}=0.
 \]
\end{cor}
\begin{proof}
 We will first prove
 \begin{equation}\label{eq: est<M}
  \lim_{\delta \downarrow 0} \qprob\pare{|\Delta \cX^\delta_T|\geq \epsilon, |\cX^\delta_T|\leq N} =0, \quad \text{ for any } \epsilon, N >0.
 \end{equation}
 To this end, for fixed $\epsilon$ and $N$, $\exp(-\Delta \cX^\delta_T)\leq e^{-\epsilon}$ when $\Delta \cX^\delta_T \geq \epsilon$. Since $U'_\delta$ converges to $U'_0$ locally uniformly, there exists a sufficiently small $\delta$ such that $e^{-\epsilon/2} \leq \fR_\delta(\cX^\delta_T) \leq e^{\epsilon/2}$ for $|\cX^\delta_T| \leq N$. On the other hand, $|\Delta X^\delta_T|\geq \epsilon/2$ when $|\Delta \xi_\delta|\leq \epsilon/2$ and $|\Delta \cX^\delta_T|\geq \epsilon$. The previous estimates combined imply that on $\{\Delta \cX^\delta_T \geq \epsilon, |\Delta \xi_\delta| \leq \epsilon/2, |\cX^\delta_T|\leq N\}$,
 \begin{equation*}\label{eq: est diff X>eps}
  \left|1-\fR_\delta(\cX^\delta_T)\exp(-\Delta \cX^\delta_T)\right||\Delta X^\delta_T| \geq (1-e^{\epsilon/2} e^{-\epsilon}) \epsilon/2 >0, \quad \text{ for sufficiently small } \delta.
 \end{equation*}
 Similarly, on $\{\Delta \cX^\delta_T \leq -\epsilon, |\Delta \xi_\delta| \leq \epsilon/2, |\cX^\delta_T|\leq N\}$,
 \begin{equation*}\label{eq: est diff X<-eps}
  \left|1-\fR_\delta(\cX^\delta_T)\exp(-\Delta \cX^\delta_T)\right||\Delta X^\delta_T| \geq (e^{-\epsilon/2} e^{\epsilon} -1) \epsilon/2 >0, \quad \text{ for sufficiently small } \delta.
 \end{equation*}
 Set $\eta =\min\{1-e^{-\epsilon/2}, e^{\epsilon/2} -1\}\cdot \epsilon/2>0$. Previous two inequalities and Lemma \ref{lem: exp est R} combined yield
 \[
  \eta \cdot \qprob\pare{|\Delta \cX^\delta_T|\geq \epsilon, |\Delta \xi_\delta|\leq \epsilon/2, |\cX^\delta_T|\leq N} \leq \expec_\qprob\bra{\left|1-\fR_\delta(\cX^\delta_T)\exp(-\Delta \cX^\delta_T)\right||\Delta X^\delta_T|} \rightarrow 0, \quad \text{ as } \delta \downarrow 0.
 \]
 Therefore \eqref{eq: est<M} follows from the previous inequality and $\lim_{\delta\downarrow 0} \qprob(|\Delta \xi_\delta|>\epsilon/2)=0$.

 Second, we will prove
 \begin{equation}\label{eq: est diff X}
  \lim_{\delta \downarrow 0} \qprob(|\Delta \cX^\delta_T|\geq \epsilon) =0.
 \end{equation}
 To this end, note that
 \begin{equation}\label{eq: qprob X ub}
 \begin{split}
  \qprob(|\cX^\delta_T|\geq N) &\leq \qprob(|\cX^\delta_T|\geq N, |\cX^0_T|\leq N/2) + \qprob(|\cX^0_T|\geq N/2)\\
  &\leq \qprob(|\Delta \cX^\delta_T|\geq N/2) + \qprob(|\cX^0_T|\geq N/2), \quad \text{ for any } N.
 \end{split}
 \end{equation}
 Let us prove in what follows
 \begin{equation}\label{eq: est q diff X>M/2}
  \lim_{\delta\downarrow 0} \qprob(|\Delta \cX^\delta_T| \geq N/2) =0, \quad \text{ for sufficiently large } N.
 \end{equation}
 Take $N/2> \max\{2, \log 1/\lR, \log \uR\}$ and set $M^\delta = N/2 \vee (|\Delta \xi_\delta| +1)$. On $\{\Delta\cX^\delta_T \leq -M^\delta\}$, $\fR_\delta(\cX^\delta_T) \exp(-\Delta \cX^\delta_T) \geq \lR \exp(N/2) >1$ and $|\Delta X^\delta_T| = |\Delta \cX^\delta_T - \Delta \xi_\delta|\geq 1$. Hence on the same set,
 \[
  \left|1-\fR_\delta(\cX^\delta_T)\exp(-\Delta \cX^\delta_T)\right||\Delta X^\delta_T|  \geq \lR \exp\pare{N/2} -1.
 \]
 On $\{\Delta \cX^\delta_T \geq M^\delta\}$, $\fR_\delta(\cX^\delta_T) \exp(-\Delta \cX^\delta_T) \leq \uR \exp(-N/2) <1$ and $|\Delta X^\delta_T|\geq 1$. Hence on the same set,
 \[
  \left|1-\fR_\delta(\cX^\delta_T)\exp(-\Delta \cX^\delta_T)\right||\Delta X^\delta_T| \geq 1-\uR \exp\pare{-N/2}.
 \]
 Set $\eta = \min\{\lR \exp(N/2) -1, 1-\uR \exp(-N/2)\}>0$. The previous two inequalities combined yield
 \begin{equation}\label{eq: est diff X>M/2}
 \begin{split}
  \eta \cdot\qprob\pare{|\Delta \cX^\delta_T| \geq M^\delta} &\leq \eta \,\expec_\qprob\bra{|\Delta X_T^\delta| \,\indic_{\{|\Delta \cX^\delta_T|\geq M^\delta\}}}\\
  & \leq \expec_\qprob\bra{\left|1- \fR_\delta(\cX^\delta_T) \exp(-\Delta \cX^\delta_T)\right||\Delta X^\delta_T| \,\indic_{\{\Delta \cX^\delta_T \geq M^\delta\}}}\\
  &\rightarrow 0, \quad \text{ as } \delta \downarrow 0,
 \end{split}
 \end{equation}
 where the convergence follows from Lemma \ref{lem: exp est R}. Therefore \eqref{eq: est q diff X>M/2} follows from
 \[
 \begin{split}
  \qprob\pare{|\Delta \cX^\delta_T| \geq N/2} &\leq \qprob\pare{|\Delta \cX^\delta_T|\geq N/2, |\Delta \xi_\delta|\leq 1} + \qprob\pare{|\Delta \xi_\delta| >1} \\
  & = \qprob\pare{|\Delta \cX^\delta_T|\geq M^\delta, |\Delta \xi_\delta|\leq 1} + \qprob\pare{|\Delta \xi_\delta| >1}\\
  & \rightarrow 0, \quad \text{ as } \delta \downarrow 0.
 \end{split}
 \]

 Switch our attention to $\qprob(|\cX^0_T|\geq N/2)$. Assumption \ref{ass: random endowment} yields $x_0 \leq \expec_\qprob[\xi_0] \leq \wt{x}_0 +\expec_\qprob[G_0 \cdot S_T] \leq \wt{x}_0$, where $G_0 \cdot S$ is a $\qprob$-local martingale bounded from below hence a $\qprob$-supermartingale. Moreover recall that $X^0$ is a $\qprob$-martingale. Therefore $\qprob(|\cX^0_T|\geq N/2)\leq 2 \,\expec_\qprob[|\cX^0_T|]/N$ which can be made arbitrarily small for sufficiently large $N$. The previous inequality combined with \eqref{eq: qprob X ub} and \eqref{eq: est q diff X>M/2} yields that $\limsup_{\delta\downarrow 0} \qprob(|\cX^\delta_T| \geq N)$ is sufficiently small for large $N$. Hence \eqref{eq: est diff X} follows from combining the previous limit superior with \eqref{eq: est<M}.

 Finally, let us prove
 \[
  \lim_{\delta \downarrow 0} \expec_\qprob\bra{\left|\Delta X^\delta_T\right|}=0.
 \]
 To this end, we have seen in \eqref{eq: est diff X>M/2} that $\lim_{\delta\downarrow 0} \expec_\qprob\bra{\left|\Delta X^\delta_T\right| \,\indic_{\{|\Delta \cX^\delta_T| \geq M^\delta\}}} =0$. On the other hand, \[\expec_\qprob\bra{|\Delta X^\delta_T| \,\indic_{\{|\Delta X^\delta_T| < M^\delta\}}} \leq \expec_\qprob\bra{|\Delta X^\delta_T| \,\indic_{\{|\Delta X^\delta_T| < M^\delta, |\Delta\xi_\delta|\leq 1\}}} + \expec_\qprob\bra{|\Delta X^\delta_T| \,\indic_{\{|\Delta X^\delta_T| < M^\delta, |\Delta\xi_\delta|> 1\}}}.\] Here the second term on the right is bounded from above by $\frac{N}{2}\qprob(|\Delta \xi_\delta| >1) + \expec_\qprob\bra{(|\Delta \xi_\delta| +1) \,\indic_{\{|\Delta\xi_\delta| >1\}}}$ which converges to $0$ as $\delta \downarrow 0$ due to Assumption \ref{ass: endowment conv}. The first term converges to $0$ as well. Indeed, since $|\Delta X^\delta_T|\leq N/2 +1$ when $|\Delta \cX^\delta_T|<M^\delta$ and $|\Delta \xi_\delta|\leq 1$, the bounded convergence theorem implies that $\lim_{\delta\downarrow 0} \expec_\qprob\bra{|\Delta X^\delta_T| \,\indic_{\{|\Delta X^\delta_T| < M^\delta, |\Delta\xi_\delta|\leq 1\}}} =0$ along any subsequence of $\delta$ such that $\Delta X^\delta_T$ converges  to $0$ $\qprob$-a.s.. Since for any sequence, there is a subsequence along which $\Delta X^\delta_T$ converges $\qprob$-a.s., the previous convergence in expectation also holds along the entire sequence of $\delta$. This argument, which combines convergence in probability with the bounded convergence theorem, will be used frequently later without mentioned explicitly.
\end{proof}

Now Theorem \ref{thm: conv R} iii) follows from Corollary \ref{cor: L^1 conv diff X} and the following result.

\begin{lem}\label{lem: Sp norm local mart}
 For any supermartingale $Z$ with $Z_0=0$,\footnote{This result holds for any probability measure which is denoted by $\prob$ in the proof.}
 \[
  \expec\bra{\sup_{0\leq t\leq T}|Z_t|^p} \leq \frac{1}{1-p} 2^p \, \expec\bra{|Z_T|}^p, \quad \text{ for any } p\in(0,1).
 \]
\end{lem}

\begin{proof}
 It follows from Doob's maximal inequality (cf. \cite[Chapter 1, Theorem 3.8]{Karatzas-Shreve-BM})that
 \[
  \lambda \, \prob\pare{\sup_{0\leq t\leq T} |Z_t| \geq \lambda} \leq 2\, \expec[|Z_T|].
 \]
 Set $Z_* = \sup_{0\leq t\leq T} |Z_t|$. It then follows
 \[
 \begin{split}
  \expec\bra{\sup_{0\leq t\leq T}|Z_t|^p} &= \expec\bra{\int_0^\infty \indic_{\{Z_* > x\}}p x^{p-1} \,dx} = \int_0^\infty \prob(Z_* > x) p x^{p-1}\, dx\\
  &\leq \int_0^\infty \min\left\{1, \frac{2\,\expec[|Z_T|]}{x}\right\} p \, x^{p-1} \, dx =\frac{1}{1-p} 2^p \, \expec[|Z_T|]^p.
 \end{split}
 \]
 Compare to the standard Doob's $L^p$-inequality where $p>1$, the only difference in proof is the last inequality.
\end{proof}

Applying the previous lemma to the $\qprob$-supermartingale $\Delta X^\delta$ and utilizing Corollary \ref{cor: L^1 conv diff X}, we obtain  $\lim_{\delta \downarrow 0} \expec_\qprob\bra{\sup_{0\leq t\leq T} |\Delta X^\delta_t|^p} =0$. Hence Theorem \ref{thm: conv R} iii) follows from Burkholder-Davis-Gundy inequality, cf. \cite[Chapter IV, Theorem 42.1]{Rogers-Williams}:
\begin{cor}\label{cor: quad-var delta X}
 If $S$ is continuous, then
 \[
  \lim_{\delta\downarrow 0} \expec_\qprob\bra{\bra{\Delta X^\delta, \Delta X^\delta}_T^{p/2}} =0, \quad \text{ for any } p\in (0,1).
 \]
\end{cor}

The following result prepares the proof of Theorem \ref{thm: conv R} ii).

\begin{lem}\label{lem: ratio exp}
 It holds that
 \[
  \lim_{\delta\downarrow 0} \frac{\expec_\prob\bra{\exp\pare{-\cX^\delta_T}}}{\expec_\prob\bra{\exp\pare{-\cX^0_T}}}=1.
 \]
\end{lem}

\begin{proof}
 Proposition \ref{prop: u_delta wellposed} implies that
 \[
 \frac{d\qprob}{d\prob} =\frac{\exp(-\cX^0_T)}{\expec_\prob[\exp(-\cX^0_T)]}.
 \]
 After changing to the measure $\qprob$, the statement is equivalent to
 \begin{equation}\label{eq: qexp exp diff X}
 \lim_{\delta \downarrow 0} \expec_\qprob[\exp(-\Delta \cX^\delta_T)]=1.
 \end{equation}

 Fix $N> \max\{C, \log 1/\lR\}$ where $C$ is the constant in Assumption \ref{ass: endowment conv}. It follows from \eqref{eq: est diff X} that
 \begin{equation}\label{eq: exp dX>-M}
  \lim_{\delta\downarrow 0} \expec_\qprob\bra{\exp(-\Delta \cX^\delta_T) \,\indic_{\{\Delta \cX^\delta_T\geq -N\}}} =1.
 \end{equation}
 On the other hand, when $\Delta \cX^\delta_T \leq -N$, $\Delta X^\delta_T = \Delta \cX^\delta_T - \Delta \xi_\delta \leq -N+C <0$, then
 \[
 \begin{split}
  \left|1-\fR_\delta (\cX^\delta_T) \exp(-\Delta \cX^\delta_T)\right| |\Delta X^\delta_T| &= \exp(-\Delta \cX^\delta_T) \left|\exp(\Delta \cX^\delta_T) - \fR_\delta(\cX^\delta_T)\right||\Delta X^\delta_T| \\
  & \geq \exp(-\Delta \cX^\delta_T) \pare{\lR - \exp(-N)} (N- C).
 \end{split}
 \]
 Set $\eta = \pare{\lR - \exp(-N)} (N- C)>0$. It then follows from Lemma \ref{lem: exp est R} that
 \begin{equation}\label{eq: exp dX<-M}
  \eta \cdot \expec_\qprob\bra{\exp(-\Delta \cX^\delta_T) \, \indic_{\{\Delta \cX^\delta_T \leq -N\}}} \leq \expec_\qprob\bra{\left|1-\fR_\delta(\cX^\delta_T) \exp(-\Delta \cX^\delta_T)\right| |\Delta X^\delta_T| \,\indic_{\{\Delta \cX^\delta_T \leq -N\}}} \rightarrow 0, \quad \text{ as } \delta\downarrow 0.
 \end{equation}
 As a result, \eqref{eq: qexp exp diff X} follows from combining \eqref{eq: exp dX>-M} and \eqref{eq: exp dX<-M}.
\end{proof}

Now we are ready to prove Theorem \ref{thm: conv R} ii).

\begin{prop}\label{prop: value fun conv}
 It holds that
 \[
  \lim_{\delta \downarrow 0} u_\delta = u_0.
 \]
\end{prop}

\begin{proof}
 After changing to the measure $\qprob$, the statement is equivalent to
 \[
  1=\lim_{\delta\downarrow 0} \frac{\expec_\prob\bra{U_\delta(\cX^\delta_T)}}{\expec_\prob\bra{U_0(\cX^0_T)}} = \lim_{\delta\downarrow 0} \expec_\qprob\bra{\frac{U_\delta(\cX^\delta_T)}{U_0(\cX^0_T)}}.
 \]
 In what follows, we will prove
 \begin{equation}\label{eq: limsup exp ratio}
  \limsup_{\delta\downarrow 0} \expec_\qprob\bra{\frac{U_\delta(\cX^\delta_T)}{U_0(\cX^0_T)}} \leq 1;
 \end{equation}
 while $\liminf_{\delta\downarrow 0} \expec_\qprob\bra{\frac{U_\delta(\cX^\delta_T)}{U_0(\cX^0_T)}} \geq 1$ can be proved similarly. To prove \eqref{eq: limsup exp ratio}, we will estimate the limit superior of the expectation on sets $\{-N \leq \cX^\delta_T \leq N\}$, $\{\cX^\delta_T >N\}$, and $\{\cX^\delta_T<-N\}$ separately, for a fixed sufficiently large $N$, in the following three steps.

 \vspace{2mm}
 \noindent \textit{\underline{Step 1: on $\{-N\leq \cX^\delta_T\leq N\}$.} } For any $\epsilon, N>0$, there exists $\delta_{\epsilon, N}$ such that $1-\epsilon \leq \frac{U'_\delta(x)}{U'_0(x)} \leq 1+\epsilon$ for $x\in (-N, N)$ and $\delta \leq \delta_{\epsilon, N}$. Integrating $U'_\delta(x) \leq (1+\epsilon)U'_0(x)$ on $(x, N)$ gives $U_\delta(x) \geq (1+\epsilon) U_0(x) - (1+\epsilon) U_0(N) + U_\delta(N)$, which yields
 \[
  \frac{U_\delta(x)}{U_0(x)} \leq 1+\epsilon + \frac{U_\delta(N) - (1+\epsilon) U_0(N)}{U_0(x)}, \quad \text{ for } x\in [-N, N] \text{ and } \delta \leq \delta_{\epsilon, N}.
 \]
 It then follows
 \begin{equation}\label{eq: est -N<X<N}
 \begin{split}
  &\expec_\qprob\bra{\frac{U_\delta (\cX^\delta_T)}{U_0(\cX^0_T)} \, \indic_{\{-N\leq \cX^\delta_T \leq N\}}} = \expec_\qprob\bra{\frac{U_\delta(\cX^\delta_T)}{U_0(\cX^\delta_T)} \exp(- \Delta \cX^\delta_T) \, \indic_{\{-N \leq \cX^\delta_T\leq N\}}}\\
  & \leq (1+\epsilon)\, \expec_\qprob\bra{\exp(-\Delta \cX^\delta_T) \, \indic_{\{-N \leq \cX^\delta_T \leq N\}}} + \pare{U_\delta(N) - (1+\epsilon) U_0(N)} \expec_\qprob\bra{\frac{\indic_{\{-N\leq \cX^\delta_T \leq N\}}}{U_0(\cX^0_T)}} \\
  &= (1+\epsilon)\, \expec_\qprob\bra{\exp(-\Delta \cX^\delta_T) \, \indic_{\{-N \leq \cX^\delta_T \leq N\}}} + \pare{U_\delta(N) - (1+\epsilon) U_0(N)} \frac{\prob(-N\leq \cX^\delta_T\leq N)}{\expec_\prob\bra{U_0(\cX^0_T)}}.
 \end{split}
 \end{equation}
 In what follows the two terms on the right side of the previous inequality will be estimated separately.

 Let us first prepare
 \begin{equation}\label{eq: conv L0}
  \qprob(-N< \cX^0_T<N) \leq \liminf_{\delta\downarrow 0} \qprob(-N \leq \cX^\delta_T \leq N) \leq \limsup_{\delta\downarrow 0} \prob(-N\leq \cX^\delta_T \leq N) \leq \qprob(-N \leq \cX^0_T \leq N).
 \end{equation}
 Indeed, for any $\epsilon$,
 \begin{align*}
  &\qprob(-N\leq \cX^\delta_T \leq N) \\
  &= \qprob(-N - \Delta \cX^\delta_T \leq \cX^0_T \leq N- \Delta \cX^\delta_T, |\Delta \cX^\delta_T|\leq \epsilon) + \qprob(-N - \Delta \cX^\delta_T \leq \cX^0_T \leq N- \Delta \cX^\delta_T, |\Delta \cX^\delta_T|> \epsilon).
 \end{align*}
 Here the second term converges to $0$ due to \eqref{eq: est diff X}, and the first term is bounded from below by $\qprob(-N+\epsilon \leq \cX^0_T \leq N-\epsilon, |\Delta \cX^\delta_T| \leq \epsilon)$ whose limit, as $\delta \downarrow 0$, is $\qprob(-N+\epsilon \leq \cX^0_T \leq N-\epsilon)$. Hence the first inequality in \eqref{eq: conv L0} follows since $\epsilon$ is chosen arbitrarily. The third inequality in \eqref{eq: conv L0} can be proved similarly.

 Now to estimate the first term on the right side of \eqref{eq: est -N<X<N}, note
 \begin{align*}
  &\expec_\qprob\bra{\exp(-\Delta \cX^\delta_T) \indic_{\{-N \leq \cX^\delta_T \leq N, \cX^0_T \leq 2N\}}} \\
  &= \expec_\qprob\bra{\pare{\exp(-\Delta \cX^\delta_T)-1} \indic_{\{-N \leq \cX^\delta_T \leq N, \cX^0_T \leq 2N\}}} + \qprob\pare{-N \leq \cX^\delta_T \leq N, \cX^0_T \leq 2N}.
 \end{align*}
 Here, since $\Delta \cX^\delta_T \geq -3N$ when $-N \leq \cX^\delta_T \leq N$ and $\cX^0_T \leq 2N$, then the first term on the right hand side converges to zero by the bounded convergence theorem and \eqref{eq: est diff X}. For the second term, we employ the same estimate as in \eqref{eq: conv L0}. Combining estimates for both terms, we obtain
 \[
 \begin{split}
  \qprob(-N<\cX^0_T < N) &\leq \liminf_{\delta \downarrow 0} \expec_\qprob\bra{\exp(-\Delta \cX^\delta_T) \indic_{\{-N\leq \cX^\delta_T\leq N, \cX^0_T \leq 2N\}}}\\
  &\leq \limsup_{\delta \downarrow 0} \expec_\qprob\bra{\exp(-\Delta \cX^\delta_T) \indic_{\{-N\leq \cX^\delta_T\leq N, \cX^0_T \leq 2N\}}} \leq \qprob(-N \leq \cX^0_T \leq N).
 \end{split}
 \]
 On the other hand, $\Delta \cX^\delta_T \leq -N$ when $-N\leq \cX^\delta_T \leq N$ and $\cX^0_T > 2N$. Therefore
 \[
  \limsup_{\delta\downarrow 0}\expec_\qprob\bra{\exp(-\Delta \cX^\delta_T) \, \indic_{\{-N\leq \cX^\delta_T\leq N, \cX^0_T > 2N\}}} \leq \lim_{\delta\downarrow 0} \expec_\qprob\bra{\exp(-\Delta \cX^\delta_T) \,\indic_{\{\Delta \cX^\delta_T \leq -N\}}} = 0, \quad \text{ as } \delta \downarrow 0,
 \]
 where the last convergence holds owing to \eqref{eq: exp dX<-M}.
 The previous two convergence combined imply
 \begin{equation}\label{eq: bet M exp}
 \begin{split}
 \qprob(-N<\cX^0_T <N)&\leq \liminf_{\delta \downarrow 0} \expec_\qprob\bra{\exp(-\Delta \cX^\delta_T) \,\indic_{\{-N \leq \cX^\delta_T \leq N\}}}\\
 &\leq
  \limsup_{\delta \downarrow 0} \expec_\qprob\bra{\exp(-\Delta \cX^\delta_T) \,\indic_{\{-N \leq \cX^\delta_T \leq N\}}} \leq \qprob(-N\leq \cX^0_T \leq N).
 \end{split}
 \end{equation}

 To estimate the second term on the right of \eqref{eq: est -N<X<N}, note $U_\delta(N) -(1+\epsilon) U_0(N)<0$, for sufficiently small $\delta$, and $\expec_{\prob}[U_0(\cX^0_T)]<0$. The third inequality in \eqref{eq: conv L0} (where $\qprob$ can be replaced by $\prob$, since $\qprob\sim \prob$) yields
 \begin{equation}\label{eq: bet M ind}
  \limsup_{\delta \downarrow 0} \pare{U_\delta(N) - (1+\epsilon) U_0(N)} \frac{\prob(-N\leq \cX^\delta_T\leq N)}{\expec_\prob[U_0(\cX^0_T)]} \leq -\epsilon\, U_0(N) \frac{\prob(-N\leq \cX^0_T \leq N)}{\expec_\prob[U_0(\cX^0_T)]}.
 \end{equation}

 \vspace{2mm}

 \noindent \textit{\underline{Step 2: on $\{\cX^\delta_T >N\}$.}} Integrating $\lR U'_0(x) \leq U'_\delta(x)$ on $(N, x)$ yields that $\lR U_0(x) - \lR U_0(N) + U_\delta(N) \leq U_\delta(x)$ for $x> N$. This implies
 \begin{equation}\label{eq: est X>N}
  \expec_\qprob\bra{\frac{U_\delta(\cX^\delta_T)}{U_0(\cX^0_T)} \,\indic_{\{\cX^\delta_T > N\}}} \leq \lR \, \expec_\qprob\bra{\exp(-\Delta \cX^\delta_T) \,\indic_{\{\cX^\delta_T > N\}}} + \pare{U_\delta(N) - \lR U_0(N)} \frac{\prob(\cX^\delta_T > N)}{\expec_\prob[U_0(\cX^0_T)]}.
 \end{equation}
 Lemma \ref{lem: ratio exp} and the first inequality in \eqref{eq: bet M exp} combined give
 \begin{equation}\label{eq: outside M exp}
  \limsup_{\delta \downarrow 0} \expec_\qprob\bra{\exp(- \Delta \cX^\delta_T) \, \indic_{\{\cX^\delta_T > N, \cX^\delta_T < -N\}}} \leq \qprob(\cX^0_T \geq N, \cX^0_T \leq -N).
 \end{equation}
 On the other hand, $\uR U_0(N) \leq U_\delta(N) \leq \lR U_0(N)<0$ for sufficiently small $\delta$. Combining the previous inequality with $U_\delta(N) - \lR U_0(N) \geq U_\delta(0) - \lR U_0(0)$, we obtain $0\geq U_\delta(N) - \lR U_0(N) \geq U_\delta(0) - \lR U_0(0)$, where the right side is bounded uniformly in $\delta$. Utilizing the similar argument as in \eqref{eq: conv L0},  we obtain $\limsup_{\delta\downarrow 0} \prob(\cX^\delta_T >N)\leq \prob(\cX^0_T \leq N)$. Combining above estimates for the right side of \eqref{eq: est X>N},
 \begin{equation}\label{eq: >M}
  \limsup_{\delta \downarrow 0} \expec_\qprob\bra{\frac{U_\delta(\cX^\delta_T)}{U_0(\cX^0_T)} \,\indic_{\{\cX^\delta_T > N\}}} \leq \lR \,\qprob(\cX^0_T \geq N, \cX^0_T \leq -N) + (1-\lR) \,U_0(0) \frac{\prob(\cX^0_T \geq N)}{\expec_\prob[U_0(\cX^0_T)]}.
 \end{equation}

 \vspace{2mm}
 \noindent \textit{\underline{Step 3: on $\{\cX^\delta_T <-N\}$.}} Integrating $U'_\delta(x) \leq \uR U'_0(x)$ on $(x, -N)$ gives $U_\delta(x) \geq \uR U_0(x) + U_\delta(-N) - \uR U_0(-N) \geq \uR U_0(x)$, where the second inequality holds since $U_\delta(-N) \geq \uR U_0(-N)$ for sufficiently small $\delta$. As a result, we have from \eqref{eq: outside M exp} that
 \begin{equation}\label{eq: <-M}
  \limsup_{\delta \downarrow 0} \expec_\qprob\bra{\frac{U_\delta(\cX^\delta_T)}{U_0(\cX^0_T)} \,\indic_{\{\cX^\delta_T < -N\}}} \leq \uR \limsup_{\delta \downarrow 0} \expec_\qprob \bra{\exp(-\Delta \cX^\delta_T) \, \indic_{\{\cX^\delta_T < -N\}}} \leq \uR \,\qprob(\cX^0_T \geq N, \cX^0_T\leq -N).
 \end{equation}

 \vspace{2mm}
 \noindent Finally combining \eqref{eq: bet M exp}, \eqref{eq: bet M ind},  \eqref{eq: >M}, and \eqref{eq: <-M}, \eqref{eq: limsup exp ratio} follows after sending $\epsilon \downarrow 0$ then $N\uparrow \infty$.
\end{proof}

\begin{proof}[Proof of Corollary \ref{cor: convergence rate}]
 Following the discussion after Lemma \ref{lem: M^a indep delta}, we consider problem \eqref{def: u_delta} for $\overline{U}_\delta(\alpha_\delta x)$ and $\overline{\xi}_\delta = \alpha_\delta x_0$. After the previous change of variable, $f(\delta) = \sup_{x\in \Real} |\overline{\fR}_\delta(x)-1|$ where $\overline{\fR}_\delta(x) = \overline{U}'_\delta(x) / \exp(-x)$. In what follows, we add a bar to random variables and processes associated to the problem for $\overline{U}_\delta$. In the rest of the proof, $C$ represents a constant which may be different in different places.

 First, we utilize the argument in Lemma \ref{lem: exp est R} to prove
 \begin{equation}\label{eq: exp conv rate}
  \expec_\qprob\bra{\left|1- \overline{\fR}_\delta(\overline{\cX}^\delta_T) \exp(-\Delta \overline{\cX}^\delta_T)\right| |\Delta \overline{X}^\delta_T|}\leq C\pare{f(\delta)^2 + g(\delta)^2 + h(\delta)}, \quad \text{ for sufficiently small } \delta.
 \end{equation}
 To this end, we have seen in Lemma \ref{lem: exp est R} that the left side is bounded from above by
 \begin{equation}\label{eq: conv rate ub}
  2 \,\expec_\qprob\bra{\left| \overline{\fR}_\delta(\overline{X}^0_T + \alpha_\delta \xi_\delta) \exp(-\Delta \overline{\xi}_\delta) -1\right| \,\left|\overline{I}_\delta\pare{\overline{U}'_0(\overline{\cX}^0_T)}-\overline{\cX}^0_T - \Delta \overline{\xi}_\delta\right|},
 \end{equation}
 where $\Delta\overline{\xi}_\delta = \alpha_\delta \xi_\delta - \xi_0$. To estimate the expectation above, note that $|\Delta \overline{\xi}_\delta| \leq  C g(\delta) + |\Delta \xi_\delta|$, where the constant $C$ depends on the uniform bound of $|\Delta \xi_\delta|$ (cf. assumptions of Corollary \ref{cor: convergence rate}). Then
 \[
  1- C\pare{g(\delta) + |\Delta \xi_\delta|} \leq \exp\pare{-Cg(\delta) - |\Delta \xi_\delta|} \leq \exp(-\Delta \overline{\xi}_\delta) \leq \exp\pare{C g(\delta) + |\Delta \xi_\delta|} \leq 1+ C\pare{g(\delta) + |\Delta \xi_\delta|},
 \]
 where the first inequality follows from $e^{-y} \geq 1- y$ for $y>0$ and the fourth inequality holds due to $e^y = 1+ \int_0^y e^z \,dz \leq 1+ C y$ when $e^y \leq C$. On the other hand, $1- f(\delta) \leq \overline{\fR}_\delta \leq 1+ f(\delta)$ for sufficiently small $\delta$.
 Therefore
 \begin{equation}\label{eq: conv rate est1}
 \begin{split}
  \left|\overline{\fR}_\delta(\overline{X}^0_T + \alpha_\delta \xi_0) \exp(-\Delta \overline{\xi}_\delta) -1\right| &\leq f(\delta) + C \pare{g(\delta)+ |\Delta \xi_\delta|} + C f(\delta) \pare{g(\delta) + |\Delta \xi_\delta|} \\
  &\leq C\pare{f(\delta) + g(\delta) + |\Delta \xi_\delta|}, \quad \qprob-a.s.,
 \end{split}
 \end{equation}
 for sufficiently small $\delta$. On the other hand, we have seen in Lemma \ref{lem: exp est R} that $\overline{I}_\delta(\overline{U}'_0(\overline{x})) -\overline{x} = \log \overline{\fR}_\delta (\overline{I}_\delta (\overline{y}))$ where $\overline{y} = \overline{U}'_0(\overline{x})$. It then follows $-2f(\delta) \leq \overline{I}_\delta(\overline{U}'_0(\overline{x})) -\overline{x} \leq 2 f(\delta)$, where we use $\log(1-y) = -\int_{-y}^0 (1+z)^{-1} dz \geq -2y$ for $0<y<1/2$ and $\log(1+y) \leq y$ for $y>0$. As a result
 \begin{equation}\label{eq: conv rate est2}
  \left|\overline{I}_\delta\pare{\overline{U}'_0(\overline{\cX}^0_T)}-\overline{\cX}^0_T - \Delta \overline{\xi}_\delta\right| \leq 2f(\delta) + C g(\delta) +|\Delta \xi_\delta|, \quad \qprob-a.s.,
 \end{equation}
 for sufficiently small $\delta$.
 Combining \eqref{eq: conv rate est1} and \eqref{eq: conv rate est2}, we obtain that the expectation in \eqref{eq: conv rate ub} is bounded from above by
 \[
  C \,\expec_\qprob\bra{(f(\delta)+ g(\delta)+|\Delta \xi_\delta|)^2} \leq C\pare{f(\delta)^2 + g(\delta)^2 + \expec_\qprob[|\Delta \xi_\delta|^2]}, \quad \text{ for sufficiently small } \delta.
 \]
 This confirms \eqref{eq: exp conv rate}.

 In the next step, we will prove
 \begin{equation}\label{eq: conv rate L1}
  \expec_\qprob\bra{|\Delta \overline{X}^\delta_T|} \leq C\pare{f(\delta)^2 + g(\delta)^2 + h(\delta)}, \quad \text{ for sufficiently small }\delta.
 \end{equation}
 Indeed, an argument similar to that in Corollary \ref{cor: L^1 conv diff X} implies that there exists $N, \eta>0$ such that
 \[
  \eta \, \expec_\qprob\bra{\left|\Delta \overline{X}^\delta_T\right| \, \indic_{\{|\Delta \overline{X}^\delta_T|\geq M^\delta\}}} \leq \expec_\qprob\bra{\left|1- \overline{\fR}_\delta(\overline{\cX}^\delta_T) \exp(-\Delta \overline{\cX}^\delta_T)\right| |\Delta \overline{X}^\delta_T|\, \indic_{\{|\Delta \overline{X}^\delta_T| \geq M^\delta\}}},
 \]
 where $M^\delta = N/2\vee (|\Delta \overline{\xi}_\delta| +1)$
 The previous inequality, combined with \eqref{eq: exp conv rate}, yields
 \[
  \expec_\qprob\bra{|\Delta \overline{X}^\delta_T| \,\indic_{\{|\Delta \overline{X}^\delta_T| \geq M^\delta\}}} \leq C\pare{f(\delta)^2 + g(\delta)^2 + h(\delta)}, \quad \text{ for sufficiently small } \delta.
 \]
 Now \eqref{eq: conv rate L1} follows after noticing $\expec_\qprob\bra{|\Delta \overline{X}^\delta_T| \indic_{\{|\Delta \overline{X}^\delta_T| \leq M^\delta\}}} \leq \expec_\qprob\bra{|\Delta \overline{X}^\delta_T| \indic_{\{|\Delta \overline{X}^\delta_T| \geq M^\delta\}}}$.

 Finally, come back to the problem before changing of variable,
 \[
 \begin{split}
  \expec_\qprob\bra{|\Delta X^\delta_T|} &\leq \frac{1}{\alpha_\delta}\expec_\qprob\bra{|\Delta \overline{X}^\delta_T|} + \frac{|\alpha_\delta -1|}{\alpha_\delta} \expec_\qprob[|X^0_T|] \\
  &\leq C\bra{f(\delta)^2+g(\delta)^2 + h(\delta) + g(\delta)}\\
  &\leq C\pare{f(\delta)^2 + g(\delta) + h(\delta)}, \quad \text{ for sufficiently small } \delta.
 \end{split}
 \]
\end{proof}

Let us now prove implications of Theorem \ref{thm: conv R} on utility-based prices.

\begin{proof}[Proof of Corollary \ref{cor: davis price}]
 Following the change of variable after Lemma \ref{lem: M^a indep delta}, we can assume without loss of generality that $\alpha_\delta=1$ for all $\delta\geq 0$ throughout this proof.
 Since $B\in L^\infty(\F_T)$, it suffices to prove
 $\lim_{\delta\downarrow 0} \expec_\qprob\bra{\left|d\qprob_\delta/d\qprob -1\right|}=0,$
 which follows from $\qprob-\lim_{\delta\downarrow 0} d\qprob_\delta/d\qprob=1$ in virtual by Scheffe's lemma.

 To prove the convergence in probability, the following form of $d\qprob_\delta/d\qprob$ can be read from Proposition \ref{prop: u_delta wellposed}:
 \begin{equation*}\label{eq: density qprob delta}
  \frac{d\qprob_\delta}{d\qprob} = \frac{U'_\delta(\cX^\delta_T)}{U'_0(\cX^0_T)} \, \frac{\expec_\prob[U'_0(\cX^0_T)]}{\expec_\prob[U'_\delta(\cX^\delta_T)]}.
 \end{equation*}
 In what follows, both factors on the right side will be proved converging to $1$.

 Let us estimate the first factor. For any given $N$ and $\epsilon$, there exists a sufficiently small $\delta$ such that $|\fR_\delta(x)-1|\leq \epsilon$ for $|x|\leq N$. Then $\qprob(|\fR_\delta(\cX^\delta_T) -1| \geq \epsilon, |\cX^\delta_T|\leq N)=0$ for sufficiently small $\delta$.
 Hence 
 \begin{align*}
 \limsup_{\delta\downarrow 0} \qprob(|\fR_\delta(\cX^\delta_T)-1|\geq \epsilon) &\leq \limsup_{\delta\downarrow 0} \qprob(|\fR_\delta(\cX^\delta_T)-1|\geq \epsilon, |\cX^\delta_T|\leq N) + \limsup_{\delta \downarrow 0} \qprob(|\cX^\delta_T|> N) \\
 &\leq \qprob(|\cX^0_T|\geq N),
 \end{align*}
 which can be made arbitrarily small for sufficiently large $N$. Therefore $\qprob-\lim_{\delta\downarrow 0} \fR_\delta(\cX^\delta_T) =1$, which combined $\qprob-\lim_{\delta\downarrow 0} \exp(-\Delta \cX^\delta_T) =1$ from \eqref{eq: est diff X}, implies
 \[
  \qprob-\lim_{\delta\downarrow 0} \frac{U'_\delta(\cX^\delta_T)}{U'_0(\cX^0_T)}=\qprob-\lim_{\delta\downarrow 0} \fR_\delta(\cX^\delta_T) \exp(-\Delta \cX^\delta_T) =1.
 \]

 In this paragraph, we will prove
 \[
  \lim_{\delta \downarrow 0} \frac{\expec_\prob[U'_\delta(\cX^\delta_T)]}{\expec_\prob[U'_0(\cX^0_T)]} =1.
 \]
 Changing to  the measure $\qprob$, the previous convergence is equivalent to
 \begin{equation}\label{eq: ratio expec mar utility}
  \lim_{\delta \downarrow 0} \expec_\qprob\bra{\frac{U'_\delta(\cX^\delta_T)}{U'_0(\cX^0_T)}}=1,
 \end{equation}
 which we will prove next. For any $\epsilon$ and $N$, there exists a sufficiently small $\delta$ such that $|\fR_\delta(\cX^\delta_T)-1|\leq \epsilon$ when $|\cX^\delta_T|\leq N$. The previous inequality combined with \eqref{eq: bet M exp} yield
 \[
 \begin{split}
  \limsup_{\delta\downarrow 0}\expec_\qprob\bra{\frac{U'_\delta(\cX^\delta_T)}{U'_0(\cX^0_T)} \, \indic_{\{|\cX^\delta_T| \leq N\}}} &= \limsup_{\delta\downarrow 0} \expec_\qprob\bra{\fR_\delta(\cX^\delta_T) \exp(-\Delta \cX^\delta_T) \, \indic_{\{|\cX^\delta_T| \leq N\}}}\\
  &\leq (1+\epsilon) \limsup_{\delta\downarrow 0}\expec_\qprob\bra{\exp(-\Delta \cX^\delta_T) \, \indic_{\{|\cX^\delta_T| \leq N\}}}\\
  &\leq (1+\epsilon) \,\qprob(|\cX^0_T| \leq N).
 \end{split}
 \]
 Similar argument also gives $\liminf_{\delta\downarrow 0} \expec_\qprob\bra{U'_\delta(\cX^\delta_T) / U'_0(\cX^0_T) \,\indic_{\{|\cX^\delta_T|\leq N\}}}\geq (1-\epsilon) \,\qprob(|\cX^0_T|<N)$.
 On the other hand, it follows from \eqref{eq: outside M exp} that
 \[
 \begin{split}
 \limsup_{\delta\downarrow 0} \expec_\qprob\bra{\fR_\delta(\cX^\delta_T) \exp(-\Delta \cX^\delta_T) \, \indic_{\{|\cX^\delta_T| >N\}}} &\leq \uR \limsup_{\delta\downarrow 0} \expec_\qprob\bra{\exp(-\Delta \cX^\delta_T) \, \indic_{\{|\cX^\delta_T| >N\}}} \\
  &\leq \uR \,\qprob(|\cX^0_T|\geq N).
 \end{split}
 \]
 Combining the previous two convergence and sending $N\uparrow \infty$ then $\epsilon\downarrow 0$, we confirm \eqref{eq: ratio expec mar utility}, hence the statement of the corollary.
\end{proof}

\begin{proof}[Proof of Corollary \ref{cor: indifference price}]
 It follows from \cite[Proposition 7.2 (i)]{Owen-Zitkovic} that $(p_\delta)_{\delta\geq 0}$ is uniformly bounded since $B\in L^{\infty}(\F_T)$. Therefore in every subsequence of $(p_\delta)_{\delta\geq 0}$ there exists a further subsequence $(p_{\delta_n})_{n\geq 0}$ converging to some limit, say $\wt{p}_0$. In the next paragraph, we will prove $\wt{p}_0 = p_0$. This implies that the entire sequence of $(p_\delta)_{\delta\geq 0}$ converges to $p_0$ as well, since the choice of subsequence is arbitrary.

 For the subsequence $(\delta_n)_{n\geq 0}$, Assumption \ref{ass: endowment conv} holds for $\xi_n = x_0+ B- p_{\delta_n}$ and $\xi_0 = x_0 + B-\wt{p}_0$ when $B$ is bounded. It then follows from Theorem \ref{thm: conv R} ii) that
 \[
  \lim_{\delta_n \downarrow 0} u_{\delta_n}(x_0+B - p_{\delta_n}) = u_0(x_0+B- \wt{p}_0).
 \]
 Apply Theorem \ref{thm: conv R} ii) with  $\xi_n= x_0$,
 \[
  \lim_{\delta_n \downarrow 0} u_{\delta_n}(x_0) = u_0(x_0).
 \]
 Since $u_{\delta_n}(x_0+B-p_{\delta_n}) = u_{\delta_n}(x_0)$, the previous two convergence combined imply $u_0(x_0+B-\wt{p}_0) = u_0(x_0)$. Then $p_0= \wt{p}_0$ follows from the uniqueness of the indifference price $p_0$.
\end{proof}

\section{Stability for utilities defined on $\Real_+$}\label{sec: stab R_+}

We will prove Theorem \ref{thm: conv R_+} in this section. To this end, we can assume without loss of generality that $D=1$ $\prob$-a.s.. Otherwise, we can define $\prob_D \sim \prob$ via $d\prob_D / d\prob= D/ \expec_\prob[D]$ and work with $\prob_D$ instead of $\prob$ throughout this section.
Assumptions \ref{ass: finite entropy}, \ref{ass: marginal bound power}, and \ref{ass: cont filtration} are enforced throughout this section, \eqref{ass: Rp conv} is satisfied as well. To simplify notation, denote $\wt{U}_p(x)= x^p/p$ and $\wt{X}^{(p)}$, $\wt{Y}^{(p)}$, and $\wt{y}_p$ quantities in Proposition \ref{prop: u_p wellposed} when $U_p$ is chosen as $\wt{U}_p$.

Denote the ratio of optimal wealth processes as
\[
 r^{(p)} = \frac{X^{(p)}}{\wt{X}^{(p)}}
\]
and introduce a sequence of auxiliary probability measures $(\prob_p)_{p<0}$ via
\[
 \frac{d\prob_p}{d\prob} = \frac{\pare{\wt{X}^{(p)}_T}^p}{\expec_\prob\bra{\pare{\wt{X}^{(p)}_T}^p}}, \quad \text{ for each } p<0.
\]
It follows from Proposition \ref{prop: u_p wellposed} that $(X^{(p)}_T)^p>0$, $\prob$-a.s., therefore $\prob_p \sim \prob$ for each $p<0$. This sequence of auxiliary measures will facilitate various estimates in this section. Another important observation is that $\wt{X}^{(p)}$ has the \emph{num\'{e}raire property} under $\prob_p$, i.e., $\expec_{\prob_p}[X_T/ \wt{X}^{(p)}_T]\leq 1$ for any admissible wealth process $X$. Indeed, Proposition \ref{prop: u_p wellposed} implies $\expec_\prob\bra{(\wt{X}^{(p)}_T)^{p-1}(X_T- \wt{X}_T^{(p)})}\leq 0$ for any admissible $X$. The claim then follows from changing the measure to $\prob_p$ in the previous inequality. As a result, every admissible wealth process $X$ deflated by $\wt{X}^{(p)}$ is a $\prob_p$-supermartingale; see \cite[Equation (3.10)]{turnpike}. In particular, $\pr$ is a $\prob_p$-supermartingale.

As the last section, we start our analysis with the following estimate.

\begin{lem}\label{lem: exp est R_+}
 It holds that
 \[
  \lim_{p\downarrow -\infty} \expec_{\prob_p}\bra{|p|\left|\fR(X^{(p)}_T) (r^{(p)}_T)^{p-1}-1\right|\left|1-r^{(p)}_T\right|} =0.
 \]
\end{lem}

\begin{proof}
 Throughout this proof we omit the superscript $(p)$ in $\pX$, $\ptX$, and $\pr$ to simplify notation. Applying Proposition \ref{prop: u_p wellposed} to $U_p$ and $\wt{U}_p$, respectively, yields
 \[
  \expec_\prob\bra{U'_p(X_T)(\wt{X}_T - X_T)} \leq 0 \quad \text{ and } \quad \expec_\prob\bra{\wt{X}^{p-1}_T (X_T - \wt{X}_T)} \leq 0.
 \]
 Summing up the previous two inequalities and changing to the measure $\prob_p$, we obtain
 \[
  \expec_{\prob_p} \bra{\pare{\frac{U'_p(X_T)}{\wt{X}^{p-1}_T}-1}\pare{1- \frac{X_T}{\wt{X}_T}}} \leq 0.
 \]
 Similar to Lemma \ref{lem: exp est R}, $(U'_p(X_T) \wt{X}^{1-p}_T-1)(1-X_T/ \wt{X}_T) \leq 0$ only when $I_p(\wt{X}^{p-1}_T)\leq X_T \leq \wt{X}_T$ or $\wt{X}_T \leq X_T \leq I_p(\wt{X}^{p-1})$, where $I_p = (U'_p)^{-1}$. In either cases,
 \[
  \pare{\pare{\frac{U'_p(X_T)}{\wt{X}^{p-1}_T}-1}\pare{1- \frac{X_T}{\wt{X}_T}}}_- \leq \pare{1-\frac{U'_p(\wt{X}_T)}{\wt{X}^{p-1}_T}}\pare{1- \frac{I_p(\wt{X}_T^{p-1})}{\wt{X}_T}}.
 \]
 Therefore,
 \[
  \expec_{\prob_p} \bra{\left|\pare{\frac{U'_p(X_T)}{\wt{X}^{p-1}_T}-1}\pare{1- \frac{X_T}{\wt{X}_T}}\right|} \leq 2 \, \expec_{\prob_p} \bra{\pare{1- \fR(\wt{X}_T)}\pare{1- \frac{I_p(\wt{X}^{p-1}_T)}{\wt{X}_T}}}.
 \]
 Note that
 \[
  \frac{I_p(x^{p-1})}{x} = \frac{I_p(y)}{y^{\frac{1}{p-1}}} = \pare{\frac{I_p(y)^{p-1}}{U'_p(I_p(y))}}^{\frac{1}{p-1}} = \fR_p(I_p(y))^{\frac{1}{1-p}},
 \]
 where $y = x^{p-1}$. Utilizing the previous identity, we obtain from the previous inequality and Assumption \ref{ass: marginal bound power} that
 \begin{equation}\label{eq: exp est power}
  \expec_{\prob_p} \bra{\left|\pare{\frac{U'_p(X_T)}{\wt{X}^{p-1}_T}-1}\pare{1- \frac{X_T}{\wt{X}_T}}\right|} \leq 2 \,\max\left\{(\uR_p -1)(\uR_p^{\frac{1}{1-p}}-1), (1-\lR_p)(1-\lR_p^{\frac{1}{1-p}})\right\}.
 \end{equation}
 Since $\limsup_{p\downarrow -\infty} |p| (\uR_p-1) <\infty$ from  \eqref{ass: Rp conv}, $\lim_{p\downarrow -\infty} \uR_p^{\frac{1}{1-p}} = \lim_{p\downarrow -\infty} \exp(\frac{1}{1-p} \log \uR_p)=1$. Therefore the first term on the right side of \eqref{eq: exp est power}, after multiplying by $|p|$, converges to $0$ as $p\downarrow -\infty$. Similar argument applies to the second term as well. As a result, the left side expectation, after multiplying $|p|$, converges to $0$ as $p\downarrow -\infty$.
\end{proof}

The previous estimate induces the convergence of $r^{(p)}_T$ in the following sense.
\begin{cor}\label{cor: r^p->1 prob}
 It holds that
 \[
  \lim_{p\downarrow -\infty} \prob_p\pare{\left|(r^{(p)}_T)^p -1\right| \geq \epsilon} =0, \quad \text{ for any } \epsilon>0.
 \]
\end{cor}
\begin{proof}
 Throughout this proof we still omit the superscript $(p)$. When $r_T^p \geq 1+\epsilon$, $1-r_T \geq 1-(1+\epsilon)^{1/p}$. Note that $(1+\epsilon)^{1/p} = \exp(p^{-1} \log(1+\epsilon)) = 1+ p^{-1}\log(1+\epsilon) + o(p^{-1})$. Hence $\lim_{p\downarrow -\infty} -p (1-(1+\epsilon)^{1/p}) = \log(1+\epsilon)>0$. Therefore when $r_T^p \geq 1+ \epsilon$, $-p (1-r_T) \geq \frac12 \log(1+\epsilon)>0$ for sufficiently small $p$. When $r^p_T\leq 1-\epsilon$, we can similarly obtain $-p (r_T-1) \geq -\frac12 \log(1-\epsilon)>0$ for sufficiently small $p$. Set $\eta= \min\{\frac12 \log(1+\epsilon), -\frac12 \log(1-\epsilon)\}>0$. The previous two estimates combined yield
 \[
  -p \,|r_T -1| \geq \eta \quad \text{ when } |r^p_T -1|\geq \epsilon \text{ for sufficiently small } p.
 \]

 On the other hand, when $r^p_T \geq 1+\epsilon$, $r_T^{p-1} \geq 1+\epsilon/2$ for sufficiently small $p$. Moreover \eqref{ass: Rp conv} and Assumption \ref{ass: marginal bound power} combined imply that $\fR_p(X_T)\geq \lR_p \geq (1+\epsilon/2)^{-\frac12}$ for sufficiently small $p$. As a result,
 \[
  \fR_p(X_T) r_T^{p-1} -1 \geq (1+\epsilon/2)^{-\frac12} (1+\epsilon/2)-1 =(1+\epsilon/2)^{\frac12} - 1>0, \quad \text{ when } r^p_T -1 \geq \epsilon \text{ for sufficiently small } p.
 \]
 Similarly,
 \[
  1- \fR(X_T) r_T^{p-1} \geq 1-(1-\epsilon/2)^{\frac12}>0, \quad \text{ when } r^p_T -1 \leq -\epsilon, \text{ for sufficiently small } p.
 \]

 Combining estimates in the last two paragraphs, we obtain
 \[
  |p| \left|\fR(X_T) r^{p-1}_T -1\right||1-r_T| \geq \eta\cdot \min\left\{\pare{1+\epsilon/2}^{\frac12} -1, 1-(1-\epsilon/2)^{\frac12}\right\}>0, \quad \text { when } |r^p_T-1|\geq \epsilon,
 \]
 for for sufficiently small $p$.
 The statement then follows from the previous inequality and Lemma \ref{lem: exp est R_+}.
\end{proof}

The previous convergence in probability implies that $(r^{(p)}_T)^p$ converges to $1$ in expectation.

\begin{prop}\label{prop: L^1 conv r^p}
 It holds that
 \[
  \lim_{p\downarrow -\infty} \expec_{\prob_p} \bra{\left|\pare{r_T^{(p)}}^p -1\right|}=0.
 \]
\end{prop}

\begin{proof}
 Throughout this proof we omit the superscript $(p)$. The proof is split into two steps. The first step proves
 \begin{equation}\label{eq: exp r^p ->1}
  \lim_{p\downarrow -\infty} \expec_{\prob_p}[r^p_T] =1.
 \end{equation}
 The second step confirms the statement.

 \vspace{2mm}
 \noindent\textit{\underline{Step 1}:} After the measure $\prob_p$ is changed to $\prob$, \eqref{eq: exp r^p ->1} is equivalent to
 \begin{equation}\label{eq: ratio power}
  \lim_{p\downarrow -\infty} \frac{\expec_\prob[X^p_T]}{\expec_\prob[\wt{X}^p_T]} =1,
 \end{equation}
 which will be proved in this step. We have seen in Proposition \ref{prop: u_p wellposed} that
 \[
  \frac{\lR_p}{y_p} \expec_\prob[X^p_T] \leq x_0 = \frac{1}{y_p} \expec_\prob\bra{U'_p(X_T) X_T} \leq \frac{\uR_p}{y_p} \expec_\prob[X^p_T],
 \]
 where Assumption \ref{ass: marginal bound power} is used to obtain two inequalities. Sending $p\downarrow -\infty$ in previous inequalities, we obtain from $\lR_p, \uR_p\rightarrow 1$,
 \[
  \lim_{p\downarrow -\infty} \frac{1}{y_p} \expec_\prob[X^p_T] =x_0.
 \]
 The optimality of $\wt{X}$ gives $\expec_\prob[X^p_T]/p \leq \expec_\prob[\wt{X}^p_T]/p= x_0 \wt{y}_p/p$. The previous convergence and $p<0$ then yields
 \[
  \limsup_{p\downarrow -\infty} \frac{\wt{y}_p}{y_p} \leq 1.
 \]
 The reverse inequality on the limit inferior will be proved in the next paragraph.

 Note that $\frac{y}{I_p(y)^{p-1}} = \frac{U'_p(x)}{x^{p-1}}$ for $x= I_p(y)$. Then Assumption \ref{ass: marginal bound power} gives $\lR_p \leq \frac{y}{I_p(y)^{p-1}} \leq \uR_p$, hence
 \[
  \lR_p^{\frac{1}{1-p}} \leq \frac{I_p(y)}{y^{\frac{1}{p-1}}} \leq \uR_p^{\frac{1}{1-p}}, \quad \text{ for } y>0.
 \]
 Proposition \ref{prop: u_p wellposed} then yields
 \[
  x_0 = \expec_\prob\bra{Y_T I_p \pare{y_p Y_T}} \leq \uR^{\frac{1}{1-p}} \expec_\prob\bra{Y_T \pare{y_p Y_T}^{\frac{1}{p-1}}} = \uR^{\frac{1}{1-p}} y_p^{\frac{1}{p-1}} \expec_\prob\bra{Y_T^q},
 \]
 where $q:= p/(p-1)$. Note $\expec_\prob[Y_T^q]^{1-p} \leq \expec_\prob[\wt{X}_T^p/ x_0^p]$ follows from $\expec_\prob[Y_T\wt{X}_T/x_0]\leq 1$ and H\"{o}lder's inequality (see e.g. \cite[Lemma 5]{Guasoni-Robertson}). The previous two inequalities combined yield $x_0 y_p \leq \uR_p \expec_\prob[\wt{X}^p_T] = \uR_p \,x_0 \wt{y}_p$. Sending $p\downarrow -\infty$ and utilizing $\lim_{p\downarrow -\infty} \uR_p=1$, we obtain from the previous inequality
 \[
  \liminf_{p\downarrow -\infty} \frac{\wt{y}_p}{y_p} \geq 1.
 \]

 Estimates from the last two paragraphs yield $\lim_{p\downarrow -\infty} y_p/\wt{y}_p =1$, which is equivalent to
 \[
  \lim_{p\downarrow -\infty} \frac{\expec_\prob\bra{U'_p(X_T) X_T}}{\expec_\prob\bra{\wt{X}_T^p}} =1.
 \]
 Since $\lR_p \expec_\prob[X_T^p] \leq \expec_\prob[U'_p(X_T) X_T] \leq \uR_p \expec_\prob[X_T^p]$, \eqref{eq: ratio power} follows from dividing by $\expec_\prob[\wt{X}_T^p]$ on both sides of the previous inequality and sending $p\downarrow -\infty$.

 \vspace{2mm}
 \noindent\textit{\underline{Step 2}:} For any $N>1$,  $\lim_{p\downarrow -\infty} \expec_{\prob_p} \bra{|r_T^p -1| \,\indic_{\{r_T^p\leq N \}}}=0$ is proved in this paragraph.
 To this end, for any $\epsilon>0$,
 \[
 \begin{split}
  \expec_{\prob_p} \bra{|r^p_T-1| \,\indic_{\{r^P_T \leq N\}}}&= \expec_{\prob_p} \bra{|r^p_T-1| \,\indic_{\{r^P_T \leq N, |r^p_T-1|\leq \epsilon\}}} + \expec_{\prob_p} \bra{|r^p_T-1| \,\indic_{\{r^P_T \leq N, |r^p_T-1|>\epsilon\}}}\\
  &\leq \epsilon + (N-1) \,\prob_p(|r^p_T-1|>\epsilon)\\
  & \rightarrow \epsilon, \quad \text{ as } p\downarrow -\infty,
 \end{split}
 \]
 where the convergence follows from Corollary \ref{cor: r^p->1 prob}. Therefore the claim is confirmed since the choice of $\epsilon$ is arbitrary in the previous inequality.

 Now $\lim_{p\downarrow -\infty} \expec_{\prob_p} \bra{|r_T^p -1| \,\indic_{\{r_T^p> N \}}}=0$ is proved in this paragraph. Combining this convergence and the one in the last paragraph confirm $\lim_{p\downarrow -\infty} \expec_{\prob_p} \bra{|r_T^p -1|}=0$. To prove the claim,
 \[
 \begin{split}
  \expec_{\prob_p} \bra{|r_T^p -1| \,\indic_{\{r_T^p> N \}}} & \leq \expec_{\prob_p} \bra{r^p_T \,\indic_{\{r^p_T > N\}}}\\
  & = \expec_{\prob_p}[r^p_T] - \expec_{\prob_p}\bra{(r^p_T-1) \,\indic_{\{r^p_T \leq N\}}} - \prob_p(r^p_T \leq N)\\
  & \rightarrow 1-0-1 =0, \quad \text{ as } p\downarrow -\infty,
 \end{split}
 \]
 where the convergence of three terms follow from the result in Step 1, the result in the last paragraph, and Corollary \ref{cor: r^p->1 prob}, respectively.
\end{proof}

The convergence of optimal payoffs in Proposition \ref{prop: L^1 conv r^p} implies the ratio of optimal wealth processes converges uniformly in probability. The proof of the following two results adapt arguments in \cite[Theorem 2.5]{Kardaras} into our context.

\begin{cor}\label{lem: r^p conv ucp}
 It holds that
 \[
  \lim_{p\downarrow -\infty} \prob_p \pare{\sup_{t\in[0,T]} \left|(\pr_T)^p-1\right|\geq \epsilon}=0.
 \]
\end{cor}

\begin{proof}
 The superscript $(p)$ is still omitted throughout to simplify notation. Recall that $r$ is a $\prob_p$-supermartingale; see the discussion before Lemma \ref{lem: exp est R_+}. Then $p<0$ implies that $r^p$ is a $\prob_p$-submartingale. Indeed,
 $\expec_{\prob_p} \bra{r^p_t \, |\, \F_s} \geq \pare{\expec_{\prob_p}\bra{r_t \,|\, \F_s}}^p \geq r_s^p$ for any $s\leq t$,
 where the Jensen's inequality is applied to obtain the first inequality.

 In the next two paragraphs, we will prove
 \begin{equation}\label{eq: sup r^p-1}
  \lim_{p\downarrow -\infty} \prob_p \pare{\left|\sup_{t\in [0,T]} r_t^p -1\right|\geq \epsilon} =0 \quad \text{ and } \quad \lim_{p\downarrow -\infty} \prob_p \pare{\left|\inf_{t\in [0,T]} r_t^p -1\right| \geq \epsilon} =0,
 \end{equation}
 for any fixed $\epsilon>0$.
 These two convergence combined confirm the statement.

 To prove the first convergence in \eqref{eq: sup r^p-1}, define $\tau_p := \inf\{t\geq 0 \,|\, r_t^p \geq 1+ \delta\}\wedge T$ for $p<0$ and $\delta>0$. It then suffices to prove
 \[
  \lim_{p\downarrow -\infty} \prob_p \pare{\tau_p < T}=0,
 \]
 since $\delta$ is arbitrarily chosen. Suppose the previous convergence does not hold. Then there exists $\eta>0$ and a subsequence, which we still denote by $\tau_p$, such that $\lim_{p\downarrow -\infty} \prob_p\pare{\tau_p < T} = \eta$. It then follows from Proposition \ref{prop: L^1 conv r^p} that
 \[
 \left|\expec_{\prob_p}\bra{r^p_T \,\indic_{\{\tau_p= T\}}} -\prob_p(\tau_p =T)\right| = \left|\expec_{\prob_p}\bra{(r^p_T-1) \,\indic_{\{\tau_p=T\}}}\right| \leq \expec_{\prob_p} [|r^p_T-1|] \rightarrow 0, \quad \text{ as } p\downarrow -\infty.
 \]
 This implies $\lim_{p\downarrow -\infty}\expec_{\prob_p}\bra{r^p_T \,\indic_{\{\tau_p=T\}}}=1-\eta$. On the other hand, the $\prob_p$-submartingale property of $r^p$ implies
 \[
  1\leq \expec_{\prob_p}[r^p_{\tau_p}] \leq \expec_{\prob_p}\bra{r^p_T} \rightarrow 1, \quad \text{ as } p\downarrow -\infty,
 \]
 where the last convergence follows from \eqref{eq: exp r^p ->1}. Hence $\lim_{p\downarrow -\infty} \expec_{\prob_p}\bra{r^p_{\tau_p}}=1$. Therefore
 \[
 \begin{split}
  1=\lim_{p\downarrow -\infty} \expec_{\prob_p}[r^p_{\tau_p}] &\geq \liminf_{p\downarrow -\infty} \expec_{\prob_p}\bra{r^p_{\tau_p} \,\indic_{\{\tau_p<T\}}} + \lim_{p\downarrow -\infty} \expec_{\prob_p}\bra{r^p_{\tau_p} \, \indic_{\{\tau_p=T\}}}\\
  & \geq (1+\delta) \eta + (1-\eta) = 1+ \delta \eta>1,
 \end{split}
 \]
 which is a contradiction.
 The proof of the second convergence in \eqref{eq: sup r^p-1} is similar.
\end{proof}

Our next goal is to pass from convergence of optimal payoffs to  convergence of optimal strategies.

\begin{prop}\label{prop: quad var diff X}
 If $S$ is continuous, then the following statements hold for any $\epsilon>0$:
 \begin{enumerate}
  \item[i)] $\lim_{p\downarrow -\infty} \prob_p \pare{\bra{ \pare{\pr}^p, \pare{\pr}^p }_{T}\geq \epsilon}=0$;
  \item[ii)] $\lim_{p\downarrow -\infty} \prob_p \pare{\bra{\cL^{(p)}, \cL^{(p)}}_T \geq \epsilon} =0$, where $\cL^{(p)} := \int_0^\cdot \pare{1/ (r^{(p)}_t)^p} d(r^{(p)}_t)^p$, i.e., $\cL^{(p)}$ is the stochastic logarithm of $(r^{(p)})^p$.
 \end{enumerate}
\end{prop}

\begin{rem}
 Under the structure condition, $[\cL^{(p)}, \cL^{(p)}]_T = \int_0^T p(\pi_p -\wt{\pi}_p)_t \,d\langle M \rangle_t\, p(\pi_p -\wt{\pi}_p)_t$, which measures how far $p(\pi_p -\wt{\pi}_p)$ is away from $0$.
\end{rem}

\begin{proof}
 The superscript $(p)$ on $r$ and $\cL$ is omitted throughout this proof. Note that $\bra{r^p, r^p}_\cdot = \int_0^\cdot |r^p|^2 \, d[\cL, \cL]_t$. Statement ii) then follows from statement i) and Corollary \ref{cor: r^p->1 prob} directly. We will prove statement i) in what follows.

 Define $\tau_p = \inf\{t\geq 0\,|\, r_t^p \geq 2\} \wedge T$. It follows from Corollary \ref{lem: r^p conv ucp} that $\lim_{p\downarrow -\infty} \prob_p \pare{\tau_p = T}=1$. Therefore it suffices to prove
 \begin{equation}\label{eq: conv local quad var}
  \lim_{p\downarrow -\infty} \prob_p\pare{\bra{r^p, r^p}_{T\wedge \tau_p} \geq \epsilon}=0.
 \end{equation}
 Set $Z^{(p)}_\cdot = r^p_{\cdot \wedge \tau_p}$. Since $r^p$ is a $\prob_p$-submartingale, so is $Z^{(p)}$. Therefore \eqref{eq: exp r^p ->1} induces $\lim_{p\downarrow -\infty} \expec_{\prob_p}[Z^{(p)}_T]=1$. On the other hand, the continuity of $S$ implies the continuity of $r^p$, hence $Z^{(p)}$ is bounded from above by $2$ for all $p<0$. The Doob-Meyer decomposition gives $Z^{(p)} = M^{(p)}+ B^{(p)}$ where $M^{(p)}$ is a $\prob_p$-martingale and $B^{(p)}$ is a continuous nondecreasing process with $B^{(p)}_0=0$. The continuity of $B^{(p)}$ follows from \cite[Theorem 1.4.14]{Karatzas-Shreve-BM}. Note $sup_{t\in[0,T]}|Z^{(p)}_t-1| \leq \sup_{t\in[0,T]} |M^{(p)}_t -1| + B^{(p)}_T$. Hence
 \[
 \begin{split}
  \expec_{\prob_p}\bra{\sup_{t\in[0,T]}|M^{(p)}_t -1|} &\leq \expec_{\prob_p}\bra{\sup_{t\in[0,T]}|Z^{(p)}_t -1|} + \expec_{\prob_p}[B^{(p)}_T]\\
  &= \expec_{\prob_p}\bra{\sup_{t\in[0,T]}|Z^{(p)}_t -1|} +\expec_{\prob_p}[Z^{(p)}_T]  - \expec_{\prob_p}[M^{(p)}_T] \\
  &\rightarrow 0 + 1- 1 =0, \quad \text{ as } p\downarrow -\infty,
 \end{split}
 \]
 where $\expec_{\prob_p}\bra{\sup_{t\in[0,T]}|Z^{(p)}_t -1|} \rightarrow 0$ holds owing to $|Z^{(p)}-1|\leq 1$ and Corollary \ref{lem: r^p conv ucp}, $\expec_{\prob_p}[M^{(p)}_T]=1$ holds because $M^{(p)}$ is a $\prob_p$-martingale. Therefore the Davis inequality yields $\lim_{p\downarrow -\infty} \expec_{\prob_p}[[M^{(p)}, M^{(p)}]^{1/2}_T] =0$, which implies $\lim_{p\downarrow -\infty} \prob_p([M^{(p)}, M^{(p)}]_T\geq \epsilon) =0$. Hence \eqref{eq: conv local quad var} is confirmed, since $B^{(p)}$ is a continuous increasing process.


\end{proof}

Last step to prove Theorem \ref{thm: conv R_+}, we are going to identify limit of $\prob_p$ as $p\downarrow -\infty$. To this end, we recall the opportunity process for power utility. The \cadlag\, semimartingale $L^{(p)}$ is called the \emph{opportunity process} for the power utility $x^p/p$ if it satisfies
\[
 L^{(p)}_t \,\frac{1}{p} \pare{X_t(\pi)}^p = \esssup{\wt{\pi}\in \mathcal{A}(\pi)} \expec_\prob\bra{\left.\frac{1}{p} \pare{X(\wt{\pi})_T}^p \right| \, \F_t},
\]
for any $t\in[0,T]$ and $\pi\in \mathcal{A}$, where $\mathcal{A}(\pi)= \{\wt{\pi} \in \mathcal{A} \,:\, \wt{\pi} = \pi \text{ on } [0,t]\}$. The existence and uniqueness of $L^{(p)}$ have been proved in \cite[Proposition 3.1]{Nutz-opp}. Thanks to the scaling property of power utility, $L^{(p)}$ can be viewed as a dynamic version of the reduced value function. In particular, the definition above implies that $L^{(p)}_0 x_0^p/p = \wt{u}_p(x_0)$, where $\wt{u}_p(x_0)$ is defined in \eqref{def: u_p} with $U_p(x) = x^p/p$, and $L_0^{(p)} x_0^{p-1} = \wt{y}_p=\wt{u}'_p(x_0)$. As a result, the density of $\prob_p$ can be rewritten as
\[
 \frac{d\prob_p}{d\prob} = \frac{\pare{\wt{y}_p \wt{Y}^{(p)}_T}^q}{p \wt{u}_p(x_0)} = \frac{\pare{L^{(p)}_0 \wt{Y}^{(p)}_T}^q}{L^{(p)}_0} = \frac{\pare{\wt{Y}^{(p)}_T}^q}{\pare{L^{(p)}_0}^{1-q}},
\]
where $q= p/(p-1)$. As $p\downarrow -\infty$, using convergence results in \cite{Nutz-asy}, we will show that the denominator in the rightmost equality above converges to $1$ and the numerator converges to the density of the minimal entropy measure $\qprob$. Therefore convergence under the sequence of measures $(\prob_p)_{p<0}$ in Proposition \ref{prop: quad var diff X} can be replaced by convergence in probability $\qprob$. This, combined with \cite[Theorem 3.2]{Nutz-asy}, concludes the proof of Theorem \ref{thm: conv R_+}.

\begin{proof}[Proof of Theorem \ref{thm: conv R_+}]
 Let us first prove
 \begin{equation}\label{eq: L^1 conv P_p}
  \lim_{p\downarrow -\infty}\expec_\prob\bra{\left|\frac{d\prob_p}{d\prob} - \frac{d\qprob}{d\prob}\right|} =0.
 \end{equation}
 To this end, when $S$ is continuous, it follows from \cite[Theorem 6.6]{Nutz-asy} that $\lim_{p\downarrow -\infty} L^{(p)}_0 = L^{\exp}_0$, where $L^{\exp}$ is the opportunity process for exponential utility $-\exp(-x)$ defined in the similar fashion as that for power utility; cf. \cite[equation (6.3)]{Nutz-asy}. Since $q\rightarrow 1$ as $p\downarrow -\infty$, then $\lim_{p\downarrow -\infty} (L_0^{(p)})^{1-q} = 1$. On the other hand, when $S$ and $(L^{(p)})_{p<0}$ are continuous, \cite[Proposition 6.13]{Nutz-asy} proved that $\wt{Y}^{(p)}$ converges in the semimartingale topology to the density of $\qprob$ as $p\downarrow -\infty$. In particular,  $\prob-\lim_{p\downarrow -\infty} \wt{Y}^{(p)}_T = d\qprob/d\prob$. Hence $\prob-\lim_{p\downarrow -\infty} (\wt{Y}^{(p)}_T)^q = d\qprob/d\prob$, which, after combined with $\lim_{p\downarrow -\infty}(L^{(p)}_0)^{1-q}=1$, implies
  \[
   \prob-\lim_{\delta\downarrow -\infty} \frac{d\prob_p}{d\prob} = \frac{d\qprob}{d\prob}.
  \]
 Hence the $\mathbb{L}^1(\prob)$ convergence in \eqref{eq: L^1 conv P_p} follows from the previous convergence and Scheffe's lemma. The assumptions on the continuity of $S$ and $(L^{(p)})_{p<0}$ are ensured by Assumption \ref{ass: cont filtration}; cf. \cite[Remark 4.2]{Nutz-asy}.

 Proposition \ref{prop: quad var diff X} ii) and \eqref{eq: L^1 conv P_p} combined yield $\qprob-\lim_{p\downarrow -\infty} \bra{p (\pi_p - \wt{\pi}_p) \cdot R}_T = 0$, where $[Z] := [Z, Z]$ is the quadratic variation for the semimartingale $Z$. Hence
 \begin{equation}\label{eq: quad conv 1}
  \prob-\lim_{p\downarrow-\infty} \bra{(1-p) (\pi_p - \wt{\pi}_p) \cdot R}_T = 0,
 \end{equation}
 since $\qprob\sim \prob$.
 On the other hand, \cite[Theorem 3.2]{Nutz-asy} proved that $(1-p) \wt{\pi}_p \rightarrow \hat{\vartheta}$ in $L^2_{loc}(M)$ as $p\downarrow -\infty$. This implies
 $\prob-\lim_{p\downarrow -\infty}[((1-p)\wt{\pi}_p - \hat{\vartheta})\cdot R]_{T\wedge \tau_n} = 0$, for a sequence of stopping time $(\tau_n)$ with $\lim_{n\uparrow \infty} \tau_n=\infty$; cf. \cite[Lemma A.3]{Nutz-asy}. The previous convergence then yields
 \begin{equation}\label{eq: quad conv 2}
  \prob-\lim_{p\downarrow -\infty}\bra{((1-p)\wt{\pi}_p - \hat{\vartheta})\cdot R}_T = 0.
 \end{equation}
 Finally, the statement is confirmed via
 \[
 \begin{split}
  \bra{((1-p)\pi_p - \hat{\vartheta})\cdot R}_T &= \bra{(1-p)(\pi_p - \wt{\pi}_p)\cdot R + ((1-p)\wt{\pi}_p - \hat{\vartheta})\cdot R}_T\\
  &\leq 2\bra{(1-p)(\pi_p - \wt{\pi}_p)\cdot R}_T + 2\bra{((1-p)\wt{\pi}_p - \hat{\vartheta})\cdot R}_T,\\
 \end{split}
 \]
 where both terms in the right side converge in probability $\prob$ to zero as we have seen in \eqref{eq: quad conv 1} and \eqref{eq: quad conv 2}.
\end{proof}

\bibliographystyle{abbrvnat}
\bibliography{biblio}

\end{document}